\documentclass[11pt,a4paper]{amsart}
\pdfoutput=1
\usepackage{graphicx}
\usepackage{stmaryrd}
\usepackage{amssymb}
\usepackage{amsmath}
\usepackage{amsthm}
\usepackage{dsfont}
\usepackage{hyperref}
\usepackage{amsfonts}
\usepackage{amsrefs}

\newtheorem{theorem}{Theorem}[section]
\newtheorem{corollary}[theorem]{Corollary}

\newtheorem{lemma}[theorem]{Lemma}

\theoremstyle{definition}

\newtheorem{remark}[theorem]{Remark}

\numberwithin{equation}{section}

\newcommand{\IP}{\mathop{\null\mathds{P}}\nolimits}

\newcommand{\IN}{\mathds{N}}

\newcommand{\IZ}{\mathds{Z}}

\newcommand{\E}{\mathbb{E}}

\newcommand{\LS}{\mathcal{L}}
\newcommand{\FLS}{S}
\newcommand{\AS}{\mathcal{A}}
\newcommand{\ALS}{\mathcal{S}}
\newcommand{\ue}{  e}
\newcommand{\uE}{  E}
\newcommand{\de}{ \vec e}
\newcommand{\dg}{ \vec g}
\newcommand{\ug}{ g}
\newcommand{\dE}{ \vec E}
\newcommand{\dWalk}{\vec \omega}
\newcommand{\uWalk}{ \omega}
\newcommand{\graph}{\mathcal{G}}

\newcommand{\Walk}{\vec \omega}		
\newcommand{\Walks}{\mathcal{W}}

\newcommand{\uLoop}{\ell}
\newcommand{\Loop}{\ell}
\newcommand{\Id}{\mathrm{Id}}
\newcommand{\overbar}[1]{\mkern 1.5mu\overline{\mkern-1.5mu#1\mkern-1.5mu}\mkern 1.5mu}

\DeclareMathOperator{\Tr}{Tr}

\begin{document}
\title[Non-backtracking loops and spin networks]{Non-backtracking loop soups and statistical mechanics on spin networks}

\author{Federico Camia}
\author{Marcin Lis}

\address{Federico Camia\\
New York University Abu Dhabi\\
Saadiyat Islands\\
Abu Dhabi\\
UAE\\
and
VU University Amsterdam\\
Department of Mathematics\\
De Boelelaan 1081a\\
1081 HV Amsterdam\\
The Netherlands}
\email{federico.camia@nyu.edu}

\address{Marcin Lis\\Mathematical Sciences \\ Chalmers University of Technology and University of Gothenburg \\
SE-41296 Gothenburg\\ Sweden}
\email{marcinl\,@\,chalmers.se}\

\begin{abstract}
We introduce and study a Markov field on the edges of a graph $\graph$ in dimension $d\geq2$ whose configurations
are spin networks. The field arises naturally as the edge-occupation field of a Poissonian model (a soup) of non-back\-tracking
loops and walks characterized by a spatial Markov property such that, conditionally on the value of the edge-occupation field on
a boundary set, the distributions of the loops and arcs on either side of the boundary are independent of each other.
The field has a Gibbs distribution with a Hamiltonian given by a sum of terms which involve only edges incident on the same vertex.
Its free energy density and other quantities can be computed exactly, and their critical behavior analyzed, in any dimension.
\end{abstract}


\maketitle

\section{Introduction}
The free energy density is an important tool and one of the main objects of study in statistical mechanics, since thermodynamic
functions can be expressed in terms of its derivatives. 
As a consequence, models whose free energy density can be computed exactly  have played a crucial
role in the development of statistical mechanics. Perhaps the main example is the two-dimensional Ising model, whose free
energy density was famously derived by Onsager. However, this situation is rare and typically restricted to dimensions one and
two, as in the case of the Ising model.

In this paper, we introduce and study a Markov field whose free energy density can be computed exactly in \emph{any dimension}.
The field arises naturally as the edge-occupation field of a stochastic model of non-backtracking loops and walks, but it does not seem to have been studied before. Such random
loop models first appeared in the work of Symanzik \cite{Sym1971} on Euclidean field theory, and can be used to prove ``triviality''
(i.e., the Gaussian nature) of Euclidean fields in dimension five and higher \cite{Aizenman1981}. A lattice version appears in the work
of Brydges, Fr\"ohlich and Spencer \cite{BFS}, who develop a random walk representation of spin systems, and use it to prove correlation
inequalities. The loops that appear in the work of Brydges, Fr\"ohlich and Spencer are allowed to backtrack.

A prototypical example of a statistical mechanical model whose partition function coincides with that of a loop model is the discrete
Gaussian free field, whose partition function can be easily expressed as the grand-canonical partition function of an ``ideal'' lattice
gas of loops, i.e., a Poissonian ensemble of lattice loops. This is a special case of the random walk representation of Brydges, Fr\"ohlich
and Spencer \cite{BFS}. In this case, the obvious triviality of the Gaussian free field is reflected in the Poissonian nature of the ensemble
of loops. In general, the ensembles of random paths and loops that appear in the work of Brydges, Fr\"ohlich and Spencer are not
Poissonian: the underlying Poisson distribution is ``tilted'' by means of a potential that is a function of the occupation field generated
by the random loops and paths on the vertices of the lattice. This suggests that the occupation field of a Poissonian ensemble of loops
is an interesting object to study. Such an occupation field is the main object of interest in this paper, although we look at edge
occupation and our loops have a non-backtracking condition that is not present in the work of Brydges, Fr\"ohlich and Spencer.

The Poissonian ensemble of loops implicit in the work of Symanzik was rediscovered by Lawler and Werner \cite{LawWer}
in connection with conformal invariance and the Schramm-Loewner Evolution (SLE), and named \emph{Brownian loop soup}.
A discrete version, called \emph{random walk loop soup} was introduced by Lawler and Trujillo Ferreras \cite{LawTru}, and studied
extensively by Le Jan \cite{Lej10,lejan11}, who discovered a connection between the vertex-occupation field of the loop soup and the
square of the discrete Gaussian free field. For a particular value of the intensity of the Poisson process, the random walk loop soup
is exactly the lattice loop model mentioned above, whose partition function coincides with that of the discrete Gaussian free field (up to a multiplicative constant).

In this paper, we introduce a special ingredient to the the loop soup recipe: a non-backtracking condition on the loops. 
It turns out that the partition function of the non-backtracking loop soup can still be linked to that of the Gaussian free field,
but the link is less direct in this case, and the proof uses a connection with the Ihara (edge) zeta function \cites{Hashimoto1989,Ihara1966,ST1996,WaFu2009}, 
as explained in Section \ref{sec:det-Ihara-GFF}.

We show that the non-backtracking loop soup possesses a \emph{spatial Markov property}, discussed in the next
section\footnote{Reference \cite{Werner15}, posted shortly after this paper, shows that the non-backtracking condition is not necessary
to have the spatial Markov property. The Markovian nature of the occupation field of the loop soup without the non-backtracking condition
was proved independently in \cite{lejan11,lejan16}.}.
As a consequence of this property, the loop soup induces a \emph{Markovian} occupation field on the edges of the lattice or graph
where it is defined, which in turn implies that the field has a Gibbs distribution. Quite remarkably, the Hamiltonian can be written explicitly and
appears to be a sum of local terms, and the partition function and the free energy density of the occupation field can all be computed exactly
in \emph{any} dimension.
These computations rely on the fact that, for this as for any loop soup, the partition function can be written in terms of a determinant.
In the translation invariant case, the relevant determinant can be calculated, and the free energy density of the field written in closed
form, on the $d$-dimensional torus for \emph{any} $d$, as mentioned at the beginning of this introduction.

Although with a different diffusion constant, the random non-backtracking walk satisfies a functional central limit theorem like the simple
random walk \cite{FH2013}. Thus, in two dimensions, in the scaling limit, the non-backtracking loop model should be related to the Brownian
loop soup of Lawler and Werner, and thus to SLE. However, in view of the role that exactly-solvable models have played in one and two dimensions,
the model we introduce in this paper may be even more interesting in higher dimensions, since it may offer some insight into critical behavior
in three and four dimensions, where conformal invariance does not play the same role as in two dimensions.

It is also worth mentioning the intriguing connection with \emph{spin networks}, a concept introduced by R. Penrose \cite{Pen71}
which has applications to quantum gravity \cite{RovSmo} and appears also in conformal and topological quantum field theory (see, e.g.,
\cite{Crane1991243, crane1991}). The non-backtracking loop soup studied in this paper generates spin network configurations with a
certain Gibbs distribution, so it can indeed be regarded as a statistical mechanical model on spin networks. 

\section{Discussion of the main results} \label{Sec:Discussion}
We introduce and study a Poissonian model of non-back\-tracking loops and arcs on graphs in dimension $d \geq 2$, and the associated
edge-occupation field, given by the number of visits to each edge by the collection of loops and arcs. Adopting the language of \cite{LawWer}
and \cite{LawlerLimic}, we refer to this model as a \emph{loop soup}, even though the ``soup'' contains in general both loops and arcs that
start and end at boundary edges. For simplicity, in this section we focus only on loops; precise and more general definitions are given in the
next section.

Consider a connected simple graph $\graph=(V, \uE)$ and associate a positive weight $x_e$ to each edge $\ue \in \uE$.
A (non-back\-tracking) loop $\uLoop$ is a closed walk on the edges of $\graph$ considered up to cyclic shifts and reversal. A loop is said to have multiplicity $m_{\uLoop}$
if it can be obtained as the concatenation of $m_{\uLoop}$ copies of the same loop, and $m_{\uLoop}$ is the largest such number. A loop $\uLoop$ with multiplicity
$m_{\uLoop}$ is assigned weight $\mu(\uLoop) = \frac{1}{m_{\uLoop}}\prod x_{\ue}$, where the product is taken over all edges traversed by $\uLoop$.
The loop soup $\LS$ studied in this paper is a
Poisson point process on the space of non-back\-tracking loops with intensity measure $\mu$.
The edge-occupation field $N_{\LS}=(N_{\ALS}(\ue))_{\ue \in E}$ induced by the loop soup is given by the total number of visits of the loops to 
each edge of $\graph$.

Our main results concern both the loop soup itself and the induced edge-occupation field.
\begin{enumerate}
\item The loop soup has a spatial Markov property such that, conditionally on the value of the edge-occupation field on a boundary set,
the distributions of the loops and arcs on either side of the boundary are independent of each other (Theorem \ref{thm:spatial_markov}).
\item The edge-occupation field has a Gibbs distribution (Theorem \ref{thm:gibbs_dist}) whose Hamiltonian can be written explicitly
(Equation \eqref{eq:hamiltonian}) and whose partition function can be expressed as
\begin{align} \nonumber
Z = \sum_{N} \prod_{v \in V} C_{v} \prod_{\ue \in \uE} \frac{x_e^{N(e)}}{N(e)!} \, ,
\end{align}
where the sum runs over all spin network \cite{Pen71} configurations $N$ and $C_{v}=C_{v}(N)$
is a suitable function of $(N(\ue))_{\ue \ni v}$.
\item The partition function can be expressed as a determinant and related to the Ihara zeta function and to the partition function of a
discrete Gaussian free field (Section \ref{sec:det-Ihara-GFF}).
\item[(4)]
In the homogeneous case, $x_{\ue} \equiv x$, on a finite $d$-regular graph, $Z<\infty$ if $x<1/(d-1)$ and $Z$ diverges for $x = 1/(d-1)$
(Corollary \ref{cor:critical-surface}).
\end{enumerate}
In the case of translation invariant weights $x_{\ue}$ on the $d$-dimensional torus, we have the following exact results.
\begin{enumerate}
\item[(5)] The partition function of the occupation field can be computed explicitly in \emph{any dimension} $d\geq2$ (Corollary \ref{cor:partition_function}).
\item[(6)] The free energy density $f$ of the occupation field on ${\mathbb Z}^d$ can be computed explicitly for \emph{any dimension}
$d\geq2$ as the thermodynamic limit of the free energy density on the $d$-dimensional torus (Corollary \ref{cor:free_energy_density}).
\end{enumerate}
In the homogeneous case, $x_{\ue}=x \; \forall \ue \in \uE$, we have the following exact results on ${\mathbb Z}^d$.
\begin{enumerate}
\item[(7)] The $d$-dimensional free energy density $f(x)$ is finite for $x<1/(2d-1)$, and has a singularity as $x \nearrow 1/(2d-1)$
which depends on the dimension $d$ (Corollary \ref{cor:singular_behavior}).
\item[(8)] The truncated two-point function decays exponentially for $x<1/(2d-1)$ (Corollary \ref{cor:exponential_decay}).
\end{enumerate}

Points (7) and (8) show that the edge-occupation field undergoes a sharp phase transition and that the critical point can be explicitly
computed and the critical behavior analyzed for periodic graphs and lattices. This is done by studying the free energy density in the
thermodynamic limit and by deriving expressions for the truncated two-point function.

In addition to the results mentioned above, Sections \ref{sec:distribution} and \ref{sec:two-point_function} contain more results on the
distribution and on the two-point function of the occupation field, respectively.

In the last section of the paper, we use the loop soup to define a spin model on the vertices of the dual graph. The spin model can be shown
to be reflection positive using to the Markov nature of the edge-occupation field. In two dimensions, we conjecture that the scaling limit of the
spin field is one of the conformal fields introduced in \cite{CGK} (see also \cite{CamLis}). We note that the same spin model generated by
an ordinary loop soup was introduced by Le Jan \cite{lejan11} and is also reflection-positive.


\section{The non-backtracking loop soup} \label{sec:non-backtracking-loop-soup}
Let $\graph=(V, \uE)$ be a connected simple graph, and let $\dE$ be the set of its directed edges. For a directed edge $\de=(t_{\de},h_{\de}) \in E$, 
$-\de=(h_{\de},t_{\de}) \in E$ is its \emph{reversal}, and $ \ue =\{t_{\de},h_{\de}\} \in  \uE$ its \emph{undirected version}. 
We assume that the graph is equipped with a positive edge weight $x_{\ue}$ for each $\ue \in E$.
By $\partial \graph$ we denote the \emph{boundary} of $\graph$, i.e.\ the (possibly empty) set of edges incident on a
vertex of degree one. 

A \emph{(non-backtracking) walk} $\dWalk$
of length $|\dWalk|=n \geq 1$ is a sequence of directed edges $\dWalk=(\dWalk_1,\ldots,\dWalk_{n+1})$
such that $t_{\dWalk_{i+1}}=h_{\dWalk_i}$ and $\dWalk_{i+1} \neq -\dWalk_i$ for $1\leq i\leq n$. 
Note that the length of a walk is the number of steps the walk makes between the edges rather than the number of edges itself.
By $\dWalk^{-1}=(-\dWalk_{|\Walk|+1},\ldots, -\dWalk_1)$ we denote its reversal, and for two walks $\dWalk$, $\dWalk'$ such that
$\dWalk_{|\dWalk|+1} =\dWalk_1\kern-2.5pt' $, we define the concatenation
\[
\dWalk \oplus \dWalk' =(\dWalk_1,\ldots,\dWalk_{|\dWalk|+1}, \dWalk_2\kern-2.5pt', \ldots, {\dWalk{\kern2pt'}}_{\kern-5pt |\dWalk'|+1}).
\]

\emph{Rooted loops} are walks starting and ending at the same directed edge. 
The \emph{multiplicity} of a rooted loop~$\dWalk$, denoted by $m_{\dWalk}$, is the largest number $m$ such that $\dWalk$ is the $m$-fold concatenation of some rooted loop $\dWalk'$ with itself.
An \emph{unoriented walk}, is a walk without a specified direction of traversal, i.e.\ a two-element equivalence class 
under the relation $\Walk \sim \Walk^{-1}$. An \emph{arc} is an unoriented walk
starting and ending on $\partial G$. Arcs will be denoted by $\alpha$.
 \emph{Unrooted loops} are equivalence classes of loops under
the cyclic shift relation $\dWalk\sim (\dWalk_{i},\dWalk_{i+1}, \ldots, \dWalk_{|\dWalk|+1} , \dWalk_{1}, \dWalk_{2}, \ldots \dWalk_{i-1})$.
Unrooted unoriented loops will be simply referred to as loops and will be denoted by~$\uLoop$. 

With a slight abuse of notation, if a function $f$ defined on walks is invariant under reversal, then $f(\alpha)$ is the evaluation of $f$ at any of the two representatives of the arc $\alpha$.
Similarly, if a function $g$ defined on rooted loops is invariant under reversal and cyclic shift, $g(\Loop)$ is the evaluation of $g$ at any representative of~$\Loop$.
The \emph{weight} of a walk $\Walk$ is
 \[
 x(\dWalk) = \prod_{i=1}^{|\dWalk|} \sqrt{x_{\uWalk_i} x_{\uWalk_{i+1}}}.
 \]
The \emph{loop} and \emph{arc measures} are given by
\[
 \mu(\uLoop) =  \frac{x(\uLoop)}{m_{\uLoop}} \qquad \text{and} \qquad \mu_{\partial}(\alpha) = x(\alpha).
\]
 By $\LS$ we will denote a realization of a Poisson point process with intensity measure $ \mu$, and
 by $\AS$ a realization of a Poisson point process with intensity measure $ \mu_{\partial}$. We will write $\mathcal{S}=\LS \cup \AS$, where $\LS$ and $\AS$ are independent.
The partition function of $\LS$ is 
\[
Z_{\LS} =\sum_{L}w(L) = \sum_{L} \prod_{\uLoop \in L} \frac{\mu(\uLoop)^{\#\uLoop}}{(\#\uLoop)!},
\]
where the sum is taken over all multi-sets $L$ of loops, called \emph{loop configurations}, and where $\#\uLoop$ is the number of occurrences of $\uLoop$ in $L$. Similarly, the partition function of $\AS$ is 
\[
Z_{\AS}= \sum_{A} w(A) = \sum_{A}   \prod_{\alpha \in A} \frac{\mu_{\partial}(\alpha)^{\#\alpha}}{ (\#\alpha)!},
\]
where the sum is taken over all multi-sets of arcs $A$, called \emph{arc configurations}. The partition function of $\mathcal{S}$ is given by $Z_{\mathcal{S}} = Z_{\LS} Z_{\AS} $.
We will only consider the cases where $Z_{\LS}Z_{\AS} <\infty$. A multi-set of loops and arcs will be called a \emph{soup configuration} or simply a \emph{configuration}. In particular, loop and arc configurations are soup configurations.

For a soup configuration $S$, we define $N_{S}=(N_S(\ue))_{\ue \in E}$ to be the \emph{network} (or \emph{edge-occupation field})
induced by $S$, i.e.\ the total number of visits of the walks from $S$ to each edge of $\graph$. $N_{S}$ is a spin network in the sense
of Penrose \cite{Pen71}, that is, if the edges incident on vertex $v$ are $e_1, \ldots, e_k$, then $\sum_{i=1}^{k} N_{S}(e_i)$ must be
even and that it cannot be smaller than $2\max_{i=1,\ldots,k} N_S(e_i)$.

One of the main results of this paper is that the distribution of the random network $N_{\ALS}$ with prescribed boundary conditions $\xi = (N(e))_{e \in \partial{\graph}}$ (in particular, the distribution of $N_{\LS}$ for zero boundary condition) 
is given by a Gibbs distribution with a local Hamiltonian.

More precisely, suppose that $e_1, \ldots, e_k$ are all the edges incident on vertex $v$ and imagine replacing edge $e_i$
by $N_S(e_i)$ distinct, colored, edges. Assume that each colored edge incident on $v$ has a unique color and let $C_v$
be the number of different ways in which those edges can be connected in such a way that each colored edge
corresponding to $e_i$ is connected to a colored edge corresponding to $e_j$ for some $j \neq i$. (It is clear
that one can always connect all colored edges in this way since, after all, $N_S(e_i)$ is the number of visits to edge
$e_i$ of the loops and arcs from the soup.) If $\graph=(V, \uE)$ is finite, we can define the Hamiltonian
${\mathcal H} = {\mathcal H}(N_S) = \sum_{v \in V} H_v$ where, for each vertex $v$ with edges $e_1, \ldots, e_k$
incident on it,
\begin{equation} \label{eq:hamiltonian}
H_v = -\log \left(\frac{C_v}{\sqrt{\prod_{i=1}^{k} N_S(e_i)!}} \prod_{i=1}^{k} x_{e_i}^{N_S(e_i)/2}\right) .
\end{equation}
\begin{theorem} [Gibbs distribution] \label{thm:gibbs_dist}
If the graph $\graph$ is finite and $\ALS$ is a soup of loops and arcs in $\graph$ such that $Z_{\mathcal S}<\infty$, then
the distribution induced on the edge-occupation field by $\ALS$ is the same as a Gibbs distribution with Hamiltonian
$\mathcal H$. That is, the edge-occupation configuration $N_S$ has probability
\[
\frac{1}{Z} e^{-\mathcal H(N_S)} ,
\]
where
\begin{align} \label{eq:part_function}
Z = \sum_{N} e^{-\mathcal H(N)} = \sum_{N} \prod_{v \in V} C_v \prod_{e \in E} \frac{x_e^{N(e)}}{N(e)!} = Z_{\mathcal S}
\end{align}
with the sums running over all network configurations $N$.
\end{theorem}

We note that the partition function \eqref{eq:part_function} is reminiscent of the random current representation used by Aizenman
\cite{Aizenman1981,Aizenman1982} (see also \cite{Aizenman2014} for a more recent application).

If $\graph$ is a trivalent graph, the combinatorics is simple and one can easily compute $C_v$. (R. Penrose uses
trivalent graphs to define spin networks precisely for this reason, although the concept is more general--see \cite{Pen71}.)
Let $e_1,e_2,e_3$ be the edges incident on some vertex $v$ of $\graph$, with occupation values $N(e_1)$, $N(e_2)$
and $N(e_3)$, and define $N_{12} = \frac{N(e_1)+N(e_2)-N(e_3)}{2}$, $N_{13} = \frac{N(e_1)+N(e_3)-N(e_2)}{2}$,
$N_{23} = \frac{N(e_2)+N(e_3)-N(e_1)}{2}$. (These are the unique solutions of the system of linear equations
$N_{12}+N_{13}=N(e_1)$, $N_{12}+N_{23}=N(e_2)$, $N_{13}+N_{23}=N(e_3)$.) It is easy to see that
$C_v = \frac{N(e_1)!N(e_2)!N(e_3)!}{N_{12}!N_{13}!N_{23}!}$ and consequently
\[
H_v = -\log \left(\frac{\sqrt{N(e_1)!N(e_2)!N(e_3)!}}{N_{12}!N_{13}!N_{23}!} x_{e_1}^{N(e_1)/2} x_{e_2}^{N(e_2)/2}
x_{e_3}^{N(e_3)/2} \right).
\]

Theorem \ref{thm:gibbs_dist} is a consequence of the \emph{conditional factorization property} of the law of $\ALS$,
which implies that the soup possesses a \emph{spatial Markov property}. In other words, the edge-occupation field
is a Markov field.
A natural way to express this property is to cut loops into arcs, and we will now make precise what we mean by that. 
Suppose that we fix a set of edges $  H \subset  \uE$. By $\graph_H$ we denote 
the modified graph $\graph$ where the edges from $H$ are cut in half, i.e.\ for each $ \ue = \{ u,v \} \in  H$,
we remove $ \ue$ from the edge-set, and add two new edges, called \emph{half-edges}, $\{u,u'\}$, $\{ v',v\}$ with $u' \neq v'$, and $u',v' \notin V$. The weight of each half-edge is
equal to the weight of the removed edge it replaces. Note that the half-edges belong to $\partial \graph_H$.
Each loop or arc in $\graph$ visiting $H$, when \emph{cut along} $H$, gives rise to a multi-set of arcs in $\graph_H$, corresponding to 
its maximal subwalks which visit only edges from $E \setminus H$, with the possible exception at the endpoints. (These 
arcs are excursions the walks make outside of $H$.)
If a walk does not visit $H$, then it is not affected when the edges of $H$ are cut in two.
Given a configuration $S$ in $\graph$, we define $S_{\graph_H}$ to be the configuration in $\graph_H$ resulting from cutting the walks
from $S$ (taken with multiplicities) along~$H$.

\begin{lemma}[Conditional factorization] \label{lem:factorization}
For a set of edges $H$, and a configuration $S$,
\[
\sum_{  T:\ T_{ \graph_H}=S_{ \graph_H} }w(T) = w( S_{\graph_H}) \prod_{  \ue \in H} N_{S}(  \ue)!
\]
\end{lemma}

\begin{proof}
Note that all $T$ with $T_{\graph_H}= S_{\graph_H}$ differ only in the way the 
arcs of $T_{\graph_H}$ are connected to each other along the edges in $H$.
It is enough to consider the case when all walks in $S$ visit $H$ and hence $S_{\graph_H}$ consists only of arcs. This is because the contribution of the loops from $T\cap T_{\graph_H}$ to $w(T)$ is the same for all $T$ with $T_{\graph_H}=S_{\graph_H}$ and is equal to the total loop contribution to $w(S_{\graph_H})$.

We will prove the result by showing that both the l.h.s. and the r.h.s. of the equality are proportional to the number
of ways one can assign colors to the visits of walks of a configuration to the edges of $H$.
In order to count that number, let $\graph_H'$ be the extended graph $\graph_H$, where for each $ \ue \in H$, we connect
the degree-one vertices of the two corresponding half-edges in $\graph_H$ with $N_{S}( \ue)$ distinguishable parallel edges (see Figure~\ref{fig}).
We can think of the new edges as being colored by $N_{S}( \ue)$ different colors.  
 
 \begin{figure}
		\begin{center}
			\includegraphics[scale=1]{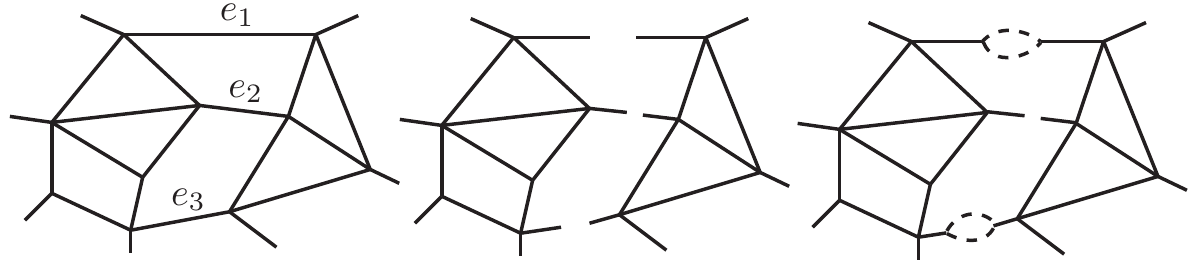}
		\end{center}
	\caption{Graphs $\graph$ with $H=\{ \ue_1,\ue_2,\ue_3\}$, $\graph_H$, and $\graph_H'$ for $S$ with $N_S(\ue_1)=2$, $N_S(\ue_2)=0$, $N_S(\ue_1)=2$}
	\label{fig}
\end{figure}
 
Consider configurations in $\graph_H'$ which visit each of the colored edges
exactly once. Each such configuration maps to a configuration in $\graph$ by forgetting the colors, i.e.\ identifying the colored edges for each $e\in H$.
Note that the resulting configuration $T$ in $\graph$ satisfies $N_T |_{H} = N_S |_{H}$. 
This mapping is many-to-one, and the cardinality of the preimage of each configuration $T$ in $G$ is proportional to $w(T)$.
We will now justify this statement. Fix a configuration $T$ in $\graph$ such that $T_{\graph_H} = S_{\graph_H}$. Consider a set $\tilde T$ where each repeated
loop or arc in $T$ is made distinguishable by adding some additional markers, and moreover where for each loop we choose a root and an orientation.
In this case, since all steps of the walks in $\tilde T$ are distinguishable, we have $N_{S}(e)!$ different ways of assigning colors to the visits of the walks in $\tilde T$ to $e \in H$. What is left to do is to account for the multiplicities coming from identifying the marked loops corresponding to the same unmarked loop and the rooted directed loops corresponding to the same unrooted undirected loop. 
This gives
\begin{align} \label{eq:main1}
\frac{\prod_{\ue \in H}N_{S}( \ue)!} {\prod_{ \alpha \in T} (\# \alpha) ! \prod_{ \Loop \in T} (\# \Loop) ! m_{\Loop}^{\# \Loop} } = \frac{\prod_{\ue \in H}N_{S}( \ue)!}{\prod_{\Loop \in T} x(\Loop)^{\#\Loop}\prod_{\alpha \in T} x(\alpha)^{\#\alpha}} w(T)
\end{align}
different colorings of the visits to $H$ of the walks in the configuration $T$.
Note that here we have used the fact that each non-backtracking walk has {two} distinct oriented versions.

Thus, the number of ways one can assign colors to the visits of walks of a configuration to the edges of $H$ can be written as
\begin{align} \label{eq:number1}
\sum_{  T:\ T_{ \graph_H}=S_{ \graph_H} } \frac{w(T)}{\prod_{\Loop \in T} x(\Loop)^{\#\Loop}\prod_{\alpha \in T} x(\alpha)^{\#\alpha}} \prod_{\ue \in H}N_{S}( \ue)!
\end{align}
 
Let us now derive a different expression for this number, this time by counting the number of ways one can connect the arcs from $S_{\graph_H}$ into loops in $\graph'_H$ in such a way that each colored edge is used only once.
To this end, take an edge $ \ue=\{u,v\}$ from $H$ and consider the corresponding two directed half-edges $\de_1=(v',v)$, $\de_2=(u',u)$, and $N_{S}( \ue)$ colored edges in $\graph_H'$. 
For the purpose of this proof we will consider \emph{directed arcs}. Again, the property of non-backtracking walks that we use is that each undirected walk has {two} distinct
directed versions. Let $\vec S_{\graph_H}$ be the multi-set of all directed versions of the arcs from $S_{\graph_H}$, and let $\vec A(\de_i)$ be the multi-set of directed arcs from $\vec S_{\graph_H}$ starting at $\de_i$, $i=1,2$. 
One has 
\[
\sum_{\vec \alpha \in \vec A(\de_i)} \# \vec \alpha = N_{S}( \ue) \qquad \text{for} \quad i=1,2,
\] where $\# \vec \alpha$ is the multiplicity of the directed arc $\vec \alpha$ in $ \vec A(\de_i)$ (which is equal to the multiplicity of its undirected version $\alpha$ in $S_{\graph_H}$).
We now distribute the $N_{S}( \ue) $ colors between the directed arcs in $\vec A(\de_i)$, $i=1,2$. Since the arcs have multiplicities, there are exactly 
\[
\frac{N_{S}( \ue)!}{\prod_{\vec \alpha_1 \in \vec A(\de_1)} (\# \vec\alpha_1)! } \times \frac{N_{S}( \ue)!}{\prod_{\vec \alpha_2 \in \vec A(\de_2)} (\# \vec\alpha_2)!}
\]
such assignments. If we take the product over $H$ and use the fact that each arc has two directed versions, we arrive at 
\[
\frac{\prod_{\ue \in H}N_{S}( \ue)!^2} {\prod_{ \alpha \in S_{\graph_H}} (\# \alpha) !^2 },
\]
where $\# \alpha$ is the multiplicity of $\alpha$ in $S_{\graph_H}$. This is the number of all possible assignments of colors to the directed arcs. We now want to forget the orientation of the arcs, so, for each arc $\alpha$, we need to pair up the opposite directed, colored arcs $\vec \alpha$ and $\vec \alpha^{-1}$.
Since we can pair any colored arc $\vec \alpha$ with any colored arc $\vec \alpha^{-1}$, we have $(\# \alpha)!$ different
pairings. Hence, 
\begin{align} \label{eq:number2}
\frac{\prod_{\ue \in H}N_{S}( \ue)!^2} {\prod_{ \alpha \in S_{\graph_H}} (\# \alpha) ! }
= \frac{w(S_{\graph_H})}{\prod_{\alpha \in S_{\graph_H}} x(\alpha)^{\#\alpha}} {\prod_{\ue \in H}N_{S}( \ue)!^2} 
\end{align}
is the number of all possible ways of connecting the arcs from $S_H$ in such a way that each colored edge is used once,
which gives us another expression for \eqref{eq:number1}.

We remind the reader that we need only consider the case when all walks in $S$ visit $H$ and hence $S_{\graph_H}$ consists only of arcs, and note that, for every $T$ such that $T_{\graph_H} = S_{\graph_H}$, we have the identity
\[
\prod_{\Loop \in T} x(\Loop)^{\#\Loop}\prod_{\alpha \in T} x(\alpha)^{\#\alpha} = \prod_{\alpha \in S_{\graph_H}} x(\alpha)^{\#\alpha} .
\]
Using this identity and comparing \eqref{eq:number1} with \eqref{eq:number2} concludes the proof of the lemma.
\end{proof}

Given a subgraph $\graph_1$ of $\graph_H$ which is a union of a number of connected components of $\graph_H$, we define $S_{\graph_1}$ to be $S_{\graph_H}$ 
restricted to walks in $\graph_1$.
If $\xi : \partial \graph \to \IN_{\geq 0}$ are boundary conditions, then we write $\IP_{\graph,\xi}$ for the probability measure governing $\ALS$ defined
on $G$ and conditioned to satisfy $N_{\ALS} |_{\partial \graph} = \xi$. 

\begin{theorem} [Spatial Markov property] \label{thm:spatial_markov}
Let $H$ be a subset of edges of $\graph$, and let $\graph_1$ be one of the connected components of $\graph_H$.  Let $\FLS$ be a configuration in $\graph$. Then,
for all boundary conditions $\xi = (N(e))_{e \in \partial{\graph}}$ on $\partial \graph$,
\begin{align*}
& \IP_{\graph,\xi} \big(\ALS_{\graph_1}= S_{\graph_1}  \big|  \ \ALS _{\graph_H \setminus \graph_1}= S_{\graph_H \setminus \graph_1}  \big)\\
& =\IP_{\graph,\xi} \big(\ALS _{\graph_1} = S_{\graph_1} \big| \  N_{\ALS} |_{\partial \graph_1}= N_S |_{\partial \graph_1}  \big) \\
&= \IP_{\graph_1,\xi_1}(\ALS = S_{\graph_1}   ), 
\end{align*}
where $\xi_1$ are boundary conditions on $\partial \graph_1$ given by 
\begin{align*}
\xi_1 = \begin{cases}
\xi & \quad \text{on} \quad \partial \graph_1 \cap \partial \graph, \\
N_{\FLS} |_{\partial \graph_1} & \quad \text{on} \quad \partial \graph_1 \setminus \partial \graph. \\
\end{cases}
\end{align*}
\end{theorem}
\begin{proof}
From Lemma~\ref{lem:factorization} and the factorization property of the Poisson point process weights, it follows that
\begin{align*}
& \IP_{\graph,\xi} \big(\ALS_{\graph_1}= S_{\graph_1}  \big|  \ \ALS _{\graph_H \setminus \graph_1}= S_{\graph_H \setminus \graph_1}  \big) \\
& = \IP_{G,\xi} \big(\ALS_{\graph_H}= S_{\graph_H}  \big) /  \IP_{G,\xi} \big(\ALS _{\graph_H \setminus \graph_1}= S_{\graph_H \setminus \graph_1} \big) \\
&= \sum_{T:T_{\graph_H}= S_{\graph_H}} w(T) \Big/ \sum_{T:T_{\graph_H \setminus \graph_1}= S_{\graph_H \setminus \graph_1}} w(T) \\
& = w(S_{\graph_1} \cup S_{\graph_H \setminus \graph_1}) \Big/ \sum_{T_1: N_{T_1} |_{\partial \graph_1} = N_{\FLS} |_{\partial \graph_1}} w(T_1 \cup S_{\graph_H \setminus \graph_1}) \\
& = w(S_{\graph_1}) / \sum_{T_1: N_{T_1} |_{\partial \graph_1} = N_{\FLS} |_{\partial \graph_1}} w(T_1) = \IP_{\graph_1,\xi_1}(\ALS = S_{\graph_1}   ),  
\end{align*}
where $T_1$ denotes a configuration in $\graph_1$.
\end{proof}

We conclude this section with a proof of Theorem \ref{thm:gibbs_dist}.
\begin{proof}[Proof of Theorem \ref{thm:gibbs_dist}]
Let $N = \{ N(e) \}_{e \in E}$ be a collection of numbers in $\IN_{\geq 0}$ indexed by the edge set $E$ of $\graph$.
We write $S \sim N$ if the soup configuration $S$ induces the occupation field $N$ (i.e., if $N_S=N$) and define
$\mathcal{S}(N) = \{ S: S \sim N \} = \{ S: N_S=N \}$. 
Let $\graph'$ be the extended graph $\graph$ where each $\ue \in E$ is replaced with $N_{S}( \ue)$ \emph{distinguishable}
parallel edges. We can think of the new edges as being colored by $N_{S}( \ue)$ different colors. Consider colored
configurations $S^*$ in $\graph'$ which visit each of the colored edges exactly once. We write $\mathcal{S}^*(N)$ for the
set of such colored configurations and $|\mathcal{S}^*(N)|$ for the cardinality of $\mathcal{S}^*(N)$.

With this notation, the partition function of the soup can be written as
\begin{eqnarray*}
Z_{\mathcal S} & = & \bigg(\sum_L \prod_{\uLoop \in L} \frac{\mu(\uLoop)^{\#\uLoop}}{(\#\uLoop)!}\bigg)
\bigg( \sum_{A}   \prod_{\alpha \in A} \frac{\mu_{\partial}(\alpha)^{\#\alpha}}{(\#\alpha)!}\bigg) \\
& = & \sum_S \prod_{\uLoop \in S} \frac{x(\uLoop)}{(\#\uLoop)!m_{\uLoop}^{\#\uLoop}} \prod_{\alpha \in S} \frac{x(\alpha)}{(\#\alpha)!} \\
& = & \sum_N \prod_{e \in E} x_e^{N(e)} \sum_{S \sim N} \prod_{\uLoop \in S} \frac{1}{(\#\uLoop)!m_{\uLoop}^{\#\uLoop}} \prod_{\alpha \in S} \frac{1}{(\#\alpha)!} \\
& = & \sum_N |{\mathcal S}^*(N)| \prod_{e \in E} \frac{x_e^{N(e)}}{N(e)!} = \sum_N \prod_{v \in V} C_v
\prod_{e \in E} \frac{x_e^{N(e)}}{N(e)!} \\
& = &  \sum_N \prod_{v \in V} e^{-H_v} = \sum_N e^{{-\mathcal H}(N)} =Z ,
\end{eqnarray*}
where we have used the obvious identity $|{\mathcal S}^*(N)| = \prod_{v \in V} C_v$ and the nontrivial identity
\[
\sum_{S \sim N} \prod_{\uLoop \in S} \frac{1}{(\#\uLoop)!m_{\uLoop}^{\#\uLoop}} \prod_{\alpha \in S} \frac{1}{(\#\alpha)!} = |{\mathcal S}^*(N)| \prod_{e \in E} \frac{1}{N(e)!} .
\]
The latter identity can be derived using arguments similar to those in the discussion preceding \eqref{eq:main1},
as explained in the lemma below.

To conclude the proof, we write
\begin{eqnarray*}
\sum_{S \sim N} \IP_{\graph}(S) & = & \frac{1}{Z_{\mathcal S}} \sum_{S \sim N} \prod_{\uLoop \in S} \frac{x(\uLoop)}{(\#\uLoop)!m_{\uLoop}^{\#\uLoop}} \prod_{\alpha \in S} \frac{x(\alpha)}{(\#\alpha)!} \\
& = & \frac{1}{Z} |{\mathcal S}^*(N)| \prod_{e \in E} \frac{x_e^{N(e)}}{N(e)!} \\
& = & \frac{1}{Z} \prod_{v \in V} C_v \prod_{e \in E} \frac{x_e^{N(e)}}{N(e)!} = \frac{1}{Z} \prod_{v \in V} e^{-H_v} \\
& = & \frac{1}{Z} e^{{-\mathcal H}(N)} . \qedhere
\end{eqnarray*}
\end{proof}

\begin{lemma}
With the notation introduced in the proof of Theorem \ref{thm:gibbs_dist}, the following identity holds:
\[
|{\mathcal S}^*(N)| = \prod_{e \in E} N(e)! \sum_{S \sim N} \prod_{\uLoop \in S} \frac{1}{(\#\uLoop)!m_{\uLoop}^{\#\uLoop}} \prod_{\alpha \in S} \frac{1}{(\#\alpha)!} .
\]
\end{lemma}
\begin{proof}
We will prove the result by counting the number of ways one can assign colors to the visits of walks of a configuration
to the edges of $\graph$. To this end, let again $\graph'$ be the extended graph $\graph$ where each $\ue \in E$ is replaced
with $N_{S}( \ue)$ \emph{distinguishable} parallel edges. We can think of the new edges as being colored by $N_{S}( \ue)$ different colors.  

Consider configurations in $\graph'$ which visit each of the colored edges exactly once. Each such colored configuration
$S^*$ maps to a configuration $S$ in $\graph$ by forgetting the colors, i.e.\ identifying the colored edges for each $e\in E$.
Note that the resulting configuration $S$ in $\graph$ satisfies $N_S = N$, and that the mapping is many-to-one. We want to
determine the cardinality of the preimage of each configuration $S$ in $G$.

Fix a configuration $S$ in $\graph$ such that $N_S = N$. Consider a set $\tilde S$ where each repeated loop or arc in $S$ is
made distinguishable by adding some additional markers, and moreover where for each loop we choose a root and an
orientation. In this case, since all steps of the walks in $\tilde S$ are distinguishable, we have $N_{S}(e)!$ different ways of
assigning colors to the visits of the walks in $\tilde S$ to $e \in E$. What is left to do is to account for the multiplicities coming
from identifying the marked loops corresponding to the same unmarked loop and the rooted directed loops corresponding to
the same unrooted undirected loop. This gives
\begin{align}
\frac{\prod_{\ue \in E}N_S( \ue)!} {\prod_{ \alpha \in S} (\# \alpha) ! \prod_{ \Loop \in S} (\# \Loop) ! m_{\Loop}^{\# \Loop} }
\end{align}
for the cardinality of the preimage of $S$. Here we again used that each non-backtracking loop has \emph{two} distinct oriented versions.
(Note that the counting doesn't work for ordinary loops without the non-backtracking condition because in that case one can have loops that
have a single oriented version instead of two--think of a loop consisting of a single edge traversed twice.)
Summing over all configurations $S \in \mathcal{S}(N)$ concludes the proof.
\end{proof}

\section{Determinantal formulas for the partition function} \label{sec:det-Ihara-GFF}
In this section we express the partition function of our model in terms of determinants of two different matrices, which,
for certain values of the edge weights, involve transition matrices of some Markov processes.
The process involved in the first determinantal formula is an asymmetric random walk on the directed edges of $\graph$,
and the one involved in the second formula is a random walk on the vertices of $\graph$. 
As a consequence of the second determinantal formula, we derive a relation between the partition function of our model
and that of the discrete Gaussian free field. 
Along the way we also uncover a connection between the partition function of our model and the {Ihara zeta function}.
In the next section we will use the first determinantal formulas to do exact computations.

We assume that $\graph$ is finite and connected and has no vertex of degree 1.
For $\de,\dg \in \dE$, let

\begin{equation} \label{eq:edge-transition-matrix}
\Lambda_{\de,\dg} = \begin{cases}
		x_e & \text{if } h_{\de}=t_{\dg} \text{ and } t_{\de} \neq h_{\dg}, \\
		0 & \text{otherwise},
		 \end{cases}
\end{equation}
and let $\rho(\Lambda)$ be the spectral radius of $\Lambda$. The next observation, a well known result, is crucial and allows for an exact solution of our model.

\begin{lemma} \label{lem:expdet} The partition function $Z_{\LS}$ is finite if and only if $\rho(\Lambda) <1$, in which case
\begin{align*} 
 {Z}_{\LS} =  {\det}^{-\frac12}(\Id - \Lambda) ,
\end{align*}
where $\Id$ is the identity matrix indexed by $\dE$.
\end{lemma}
\begin{proof}

Let $\lambda_i $, $i= 1,2,\dots,|\dE|$, be the eigenvalues of $\Lambda$, and let $\Walks^n_{\de,\de}$ be the set of all oriented walks of length $n$ starting and ending at the oriented edge~$\de$. 
Using the definition of the loop measure, we have that
for each $n$,
\[
 \sum_{\uLoop: |\uLoop| = n} \mu (\uLoop) =\sum_{\uLoop: |\uLoop|=n} \frac{x(\uLoop)}{ m_{\uLoop} } = \sum_{\de \in \dE }    \sum_{ \dWalk  \in \Walks^n_{\de,\de}}\frac{x(\dWalk)}{2n} 
 =   \frac{\Tr \Lambda^n}{2n} = \sum_{i=1}^{|\dE|} \frac{\lambda_i^n}{2n}.
\]
It follows that the first expression is summable over $n$ if and only if $\max_{i}|\lambda_i| = \rho(\Lambda) < 1$.
From the definition of the partition function, we have that
\begin{align}\label{eq:exppartition}
{Z}_{\LS} = \exp \Bigg(\sum_{n=1}^{\infty} \sum_{\uLoop: |\uLoop| =n } \mu (\uLoop)\Bigg),
\end{align}
and hence the first part of the lemma follows. To finish the proof, we notice that
\[
-2 \ln {Z}_{\LS} = -\sum_{n=1}^{\infty} \sum_{i=1}^{|\dE|}   \frac{\lambda_i^n}{n} =  \sum_{i=1}^{|\dE|} \ln(1-\lambda_i) = \ln \prod_{i=1}^{|\dE|} (1-\lambda_i) =\ln {\det}(\Id - \Lambda). \qedhere
\]
\end{proof}

In the next section we will use Lemma \ref{lem:expdet} to perform exact computations. First, however, we describe an interesting identity providing
a different determinantal representation for the partition function. Such a determinantal representation was introduced in connection with the
\emph{Ihara (edge) zeta function} \cites{Hashimoto1989,Ihara1966,ST1996}, which is defined as the infinite product
\[
\zeta(\mathbf{x}) = \prod_{\dWalk \in {\mathcal P}} (1-x(\dWalk))^{-1},
\]
where $\mathbf{x}=(x_{\ue})_{\ue \in E}$ denotes the vector of edge weights and $\mathcal P$ is the collection of all oriented, unrooted,
non-backtracking loops with multiplicity 1.  Note that in the general definition of the Ihara zeta function, the weights $x_{\ue}$ can be complex
and can be different for the two opposite orientations of an undirected edge.
By grouping together loops in \eqref{eq:exppartition} that are multiples of the same loop $\dWalk \in \mathcal P$, we easily obtain that
${Z}_{\LS} = \zeta^{\frac12}(\mathbf{x})$, and hence by Lemma~\ref{lem:expdet}, $\zeta^{-1}(\mathbf{x}) =\det(\Id-\Lambda)$.
(This is the content of Theorem 3 of \cite{ST1996} and Theorem 3.3 of \cite{HST2006}.)

To state the alternative determinantal formula, we need to first define two matrices indexed by the vertices of $\graph$.
If $\|\mathbf{x}\|_{\infty} <1$, let $D$ be a diagonal matrix with entries
\[
D_{v,v} = \sum_{\ue \ni v} \frac{x^2_{\ue}}{1-x^2_{\ue}},
\]
and let
\[
A_{v,u} = \begin{cases}
		\frac{x_{\ue}}{1-x^2_{\ue}} & \text{if } e=\{u,v\}\in E, \\
		0 & \text{otherwise}
		 \end{cases}
\]
be a weighted adjacency matrix.

\begin{lemma} \label{thm:vertexmatrix}
For $\|\mathbf{x}\|_{\infty}$ small enough,
\[
{Z}_{\LS}=  {\det}^{-1/2}(\Id + D - A)\prod_{\ue \in \uE} (1-x^2_{\ue})^{-1/2},
\]
where $\Id$ is the identity matrix indexed by $V$.
\end{lemma}
\begin{proof}
This result follows from Theorem 2 of \cite{WaFu2009} and Lemma~\ref{lem:expdet}.
\end{proof}

We will use Lemma \ref{thm:vertexmatrix} to relate the partition function of our model to that of a discrete Gaussian free field.
To this end, attach a weight $c_{\ue} \geq 0$ to each edge $\ue \in \uE$ and a \emph{killing rate} $k_v \geq 0$ to each vertex
$v \in V$, and let $\lambda_v = \sum_{\ue \ni v} c_{\ue} + k_v$. 
The weights and killing rates induce a sub-Markovian transition matrix $P$ on the vertices of $\graph$ with transition probabilities
$p_{v,u} = \frac{c_{\ue}}{\lambda_v}$ if $e = \{v,u \} \in E$, and $0$ otherwise. The transition matrix is $\lambda$-symmetric,
meaning that $\lambda_v p_{v,u} = \lambda_u p_{u,v}$. One can introduce a ``cemetery'' state $\Delta$ and extend $P$ to a
Markovian transition matrix $\bar P$ on $V \cup \Delta$ by setting $p_{v,\Delta} = \frac{k_v}{\lambda_v}$ and $p_{\Delta,\Delta}=1$.
We assume that the Green's function $G_{\bar P}(v,u)$ of the Markov chain associated to $\bar P$ (i.e., the expected number of visits
to $u$ of the chain started at $v$) is finite. This condition is equivalent to saying that $\rho(P) <1$ or that the Markov chain is always
absorbed in $\Delta$.

The \emph{discrete Gaussian free field} on $\graph$ associated with the transition matrix $\bar P$ is a collection $(\phi_v)_{v \in V}$
of mean-zero Gaussian random variables with covariance ${\mathbb E}(\phi_ v \phi_u) = G_{\bar P}(v,u)$.
It has a Gibbs distribution $e^{-H^{GFF}_{\bar P}}/Z^{GFF}_{\bar P}$ with Hamiltonian
\[
H^{GFF}_{\bar P}(\varphi) = \frac{1}{2} \sum_{\{v,u\} \in E} c_{\{v,u\}} (\varphi_{v}-\varphi_{u})^2
+ \frac{1}{2} \sum_{v \in V} k_{v} \varphi^2_v \, ,
\]
from which it follows that its partition function, $Z^{GFF}_{\bar P}$, is given by a Gaussian integral and can therefore be represented in terms
of a determinant, as follows:
\begin{equation} \label{eq:GGF-partition-function}
Z^{GFF}_{\bar P} = \prod_{v \in V}\Big(\frac{2\pi}{\lambda_v}\Big)^{1/2} {\det}^{-\frac12}(\Id - P)
\end{equation}
where $\Id$ is the identity matrix indexed by $V$.

\begin{theorem} \label{thm:GFF}
Let $\mathbf{x}$ be a vector of edge weights such that $\|\mathbf{x}\|_{\infty} <1$ and
\begin{align} \label{eq:vertexcond}
\sum_{e\ni v} \frac{x_e}{1+x_e} \leq 1 \quad \forall v \in V.
\end{align}
For $v\in V$, let
\[
\lambda^*_v = 1 + \sum_{\ue \ni v} \frac{x^2_{\ue}}{1-x^2_{\ue}} .
\]
Then, the matrix 
\[
P^*_{v,u} = \begin{cases}
		\frac{1}{\lambda^*_v} {\frac{x_{\ue}}{1-x^2_{\ue}}}& \text{if } \ue =\{v,u \} \in E, \\
		0 & \text{otherwise}
		 \end{cases}
\]
is a $\lambda^*$-symmetric, sub-Markovian transition matrix on $V$. Moreover, if ${Z}_{\LS}<\infty$,
\begin{equation} \label{eq:GFF}
{Z}_{\LS} = (2\pi)^{-|V|/2} \, \prod_{\ue \in \uE} (1-x^2_{\ue})^{-1/2} \, Z^{GFF}_{\bar P^*} ,
\end{equation}
where $\bar P^*$ is the Markovian extension of $P^*$. 
\end{theorem}
\begin{proof}
The matrix $P^*$ is obviously $\lambda^*$-symmetric and an easy computation shows that the inequalities
\[
\sum_{e \ni v} \frac{x_e}{1+x_e} \leq 1 \qquad \text{and} \qquad \sum_{\ue \ni v} \frac{1}{\lambda^*_v} \frac{x_{\ue}}{1-x^2_{\ue}} \leq 1
\]
are equivalent. Hence, condition \eqref{eq:vertexcond} guarantees that $P^*$ is sub-Markovian.
Using Lemma~\ref{thm:vertexmatrix}, for $\| \mathbf{x}\|_{\infty}$ sufficiently small, we obtain
\begin{align*}
({Z}_{\LS})^{-2} & = \prod_{\ue \in \uE} (1-x^2_{\ue}) \det(\Id + D - A) \\
& = \prod_{\ue \in \uE} (1-x^2_{\ue}) \prod_{v \in V}\Big(1 + \sum_{\ue \ni v} \frac{x^2_{\ue}}{1-x^2_{\ue}}\Big)
\det(\Id - P^*) \, .
\end{align*}
By Lemma~\ref{lem:expdet}, the l.h.s.\ and the r.h.s.\ are both polynomials, so they are equal for all $\mathbf{x}$.
Hence,
\[
{Z}_{\LS} = \prod_{\ue \in \uE} (1-x^2_{\ue})^{-1/2} \prod_{v \in V}\Big(1 + \sum_{\ue \ni v} \frac{x^2_{\ue}}{1-x^2_{\ue}}\Big)^{-1/2}
{\det}^{-\frac12}(\Id - P^*)
\]
for all $\mathbf{x}$ such that $0<{Z}_{\LS}<\infty$.
Using the definition of $\lambda^*_v$ and the determinantal formula \ref{eq:GGF-partition-function} concludes the proof.

\end{proof}

\begin{corollary} \label{cor:critical-surface}
Let $\mathbf{x}$ be a vector of edge weights such that $\|\mathbf{x}\|_{\infty} <1$ and
\[
\sum_{e\ni v} \frac{x_e}{1+x_e} \leq 1 \quad \forall v \in V.
\]
If $\exists v \in V$ such that
\[
\sum_{e\ni v} \frac{x_e}{1+x_e} < 1,
\]
then ${Z}_{\LS}<\infty$ and Equation \eqref{eq:GFF} holds. If
\[
\sum_{e\ni v} \frac{x_e}{1+x_e} = 1 \quad \forall v \in V ,
\]
then ${Z}_{\LS}$ diverges.

In particular, in the homogeneous case, $x_{\ue} \equiv x$, on a $d$-regular graph, ${Z}_{\LS}<\infty$ if $x < 1/(d-1)$ and
${Z}_{\LS}$ diverges if $x = 1/(d-1)$.
\end{corollary}
\begin{proof}
It is easy to verify that, if $\sum_{e\ni v} \frac{x_e}{1+x_e} \leq 1$ for all $v \in V$ and moreover $\sum_{e\ni v} \frac{x_e}{1+x_e} < 1$
for at least one $v \in V$, then ${Z}_{\LS}<\infty$ and $\rho(P^*)<1$. If, in addition, $\|\mathbf{x}\|_{\infty} <1$, then Theorem \ref{thm:GFF}
and the determinantal formula \eqref{eq:GGF-partition-function} hold. It is also easy to check that, if $\sum_{e\ni v} \frac{x_e}{1+x_e} = 1$ for all
$v \in V$, the transition matrix $P^*$ is Markovian and ${Z}_{\LS}$ diverges.
\end{proof}

\begin{remark}
A direct way of proving the homogeneous case result is by counting all possible non-backtracking walks on the graph, or by noticing that
$\rho(\Lambda)\leq \|\Lambda\|_1 =x(d-1)$ and using Lemma~\ref{lem:expdet}.
Similarly, for any graph with vertices of degree $d_{v} \geq 2$, one can show that a sufficient condition to ensure that ${Z}_{\LS}<\infty$ is:
$x_{\ue}<1/(d_v-1)$ for all $v \in V$ and all $\ue \ni v$. \end{remark}

We conclude this section with an interesting observation. As mentioned in the introduction, the partition function of the Gaussian free field
in the form \eqref{eq:GGF-partition-function} can be easily expressed in terms of the partition function of a ``regular'' random walk loop soup,
i.e., \emph{without} the non-backtracking condition. (The proof of this fact is completely analogous to the proof of Lemma \ref{lem:expdet}.
The interested reader is referred to the proof of Lemma 1.2 of \cite{BFS}.) This provides a link between the partition function of the
non-backtracking loop soup with transition probabilities given by $\Lambda$ and the random walk loop soup with transition probabilities
given by $P^*$. It is possible that the connection between those two loop soups and their occupation fields is realized at a deeper level
than that of the partition functions. Such a deeper connection, if it exists, could be revealed by an analysis of the determinantal formula
for the Ihara zeta function given in Theorem 2 of \cite{WaFu2009}.

\section{Exact computations}
In this section we will compute explicitly the free energy density of translation invariant models and the one-point function of homogeneous models on the torus $\IZ ^d / (n\IZ)^{d}$ for any $n \geq 1$, and $d\geq 1$, 
and hence after taking the thermodynamic limit, on all hypercubic lattices $\IZ^d$, $d\geq 1$. 
We will also prove the exponential decay of the two-point function in the subcritical regime.
The reason why an exact solution is available is the fact that the partition function of the model is given by the square root of the determinant of a matrix 
(a situation similar to the one of the Ising and dimer model, and also the discrete Gaussian free field). It follows that 
the relevant quantities are expressed in terms of the eigenvalues and can be computed for periodic graphs using the Fourier transform.
Unlike in the Ising and dimer model case,
the determinantal formulas are valid in all dimensions, similarly to the case of the discrete Gaussian free field.

\subsection{Partition function and free energy density}

Let $\graph^{d}_n = \IZ ^d / (n\IZ)^{d}$ be a $d$-dimensional torus of size $n^d$.
For a vertex $\mathbf{k}=(k_1,k_2,\ldots,k_n) \in\graph^d_n $, and a unit direction vector $\mathbf{v}=(v_1,v_2,\ldots,v_n)$ such that $v_j = \pm 1$ for some $j$,
and  $v_i = 0$ for $i \neq j$, we will write $(\mathbf{k}, \mathbf{v})$ for the directed edge $\de$ with
$t_{\de}= \mathbf{k}$ and $h_{\de} =\mathbf{k} +\mathbf{v}$. 
We consider a translation invariant weight vector $\mathbf{x}=(x_1,\ldots,x_d)$ that assigns weight $x_j$ to undirected versions of all edges of the form
$(\mathbf{k} , \mathbf{v})$, where $\mathbf{v}$ is a unit vector in the $j$-th direction. In this case, the transition matrix \eqref{eq:edge-transition-matrix} is given by
\begin{align}\label{eq:torusmatrix}
[\Lambda^d_n]_{(\mathbf{k}, \mathbf{v}),(\mathbf{l}, \mathbf{u})} = \begin{cases}
		x_j & \text{if } \mathbf{k} + \mathbf{v} =\mathbf{l} \text{ and } \mathbf{v} \neq - \mathbf{u} , \\
		0 & \text{otherwise}.
		 \end{cases}
\end{align}
Consider its Fourier transform
\[
[\hat \Lambda^d_n]_{(\mathbf{p}, \mathbf{v}),(\mathbf{q}, \mathbf{u})} = \frac{1}{n^d} \sum_{\mathbf{k}, \mathbf{l} \in\graph^d_n } e^{-\frac{2\pi i}{n}( \mathbf{p} \cdot \mathbf{k} - \mathbf{q} \cdot \mathbf{l})} [\Lambda^d_n]_{(\mathbf{k}, \mathbf{v}),(\mathbf{l}, \mathbf{u})}.
\]
It is easily seen that 
\[
[\hat \Lambda^d_n]_{(\mathbf{p}, \mathbf{v}),(\mathbf{q}, \mathbf{u})} = \begin{cases}
		x_j e^{\frac{2\pi i}{n} \mathbf{p} \cdot \mathbf{v}}& \text{if }  \mathbf{p} = \mathbf{q} \text{ and } \mathbf{v} \neq - \mathbf{u} , \\
		0 & \text{otherwise}.
		 \end{cases}
\]
Note that $\hat \Lambda^d_n({\mathbf{p}})$ is block-diagonal with blocks indexed by the vertices $\mathbf{p} \in \graph^d_n$ whose rows and columns correspond 
to the directed edges $\de$ satisfying $t_{\de} = \mathbf{p}$.
Let $\hat \Lambda^d_n({\mathbf{p}})$ denote the $2d \times 2d$ block corresponding to $\mathbf{p} \in \graph^d_n$. 
Since $\hat \Lambda^d_n$ is similar to $\Lambda^d_n$, one has that 
\begin{align} \label{eq:Fourier}
\det(\Id - \Lambda^d_n) = \det(\Id-\hat \Lambda^d_n) = \prod_{\mathbf{p} \in \graph^d_n}\det(\Id - \hat \Lambda^d_n(\mathbf{p})),
\end{align}
where $\Id$ is the $2d$-dimensional identity. One can explicitly compute these determinants,
as shown below.

\begin{lemma} \label{lem:blockdet}
We have that
\begin{align*} 
\det(\Id - \hat \Lambda^d_n(\mathbf{p}) ) = \Bigg[1+2 \sum_{i=1}^{d}\frac{x_i \big(x_i- \cos\big(\frac{2\pi}{n} p_i  \big)\big) }{1-x_i^2}\Bigg]\prod_{i=1}^d(1-x_i^2).
\end{align*}
\end{lemma}
\begin{proof}
The proof is by induction on $d$. Let $z_j=e^{\frac{2\pi i}{n} {p}_j}$. For $d=1$, 
\[ 
\Id - \hat \Lambda^1_n(\mathbf{p}) = \left( 
\begin{array}{cc}
1-x_1 \bar z_1 & 0  \\
0 & 1-x_1 z_1
\end{array} 
\right),
\]
and the statement is true. Assume that it holds true for $d \leq k$, and consider
the matrix $\Lambda^{k+1}_n(\mathbf{p})$ for $\mathbf{p} \in \graph^{k+1}_n$. 
Let $\mathbf{p}'=(p_1,p_2,\ldots,p_{k}) \in  \graph^{k}_n$ be the restriction of $\mathbf{p}$ to the first $k$ coordinates. 
For a number $a$, we will write $\underline a$ for a row- or column-vector with entries all equal to $a$. 
Let $M^{d}_n(\mathbf{p})$ be the matrix $\Id- \hat \Lambda^{d}_n(\mathbf{p})$ where each row corresponding to a directed edge $(\mathbf{p},\mathbf{v})$
 is divided by $-x_j e^{\frac{2\pi i}{n} \mathbf{p} \cdot \mathbf{v}}$, where $\mathbf{p} \cdot \mathbf{v}$ is either $p_j$ or $-p_j$. Note that $\det(\Id- \hat \Lambda^{d}_n(\mathbf{p})) =\det M^{d}_n(\mathbf{p})\prod_{i=1}^d x_i^2$.
Using linearity of the determinant, we have
 \begin{align*} \det& M^{k+1}_n(\mathbf{p})=
 \begin{vmatrix}
1-\frac{z_{k+1}}{x_{k+1}} & 0  & \underline 1 \\
0 & 1-\frac{\bar{z}_{k+1}}{x_{k+1}} & \underline 1 \\
 \underline 1 & \underline 1 &  M^{k}_n(\mathbf{p}') 
\end{vmatrix} \\
&= 
 \begin{vmatrix}
1-\frac{z_{k+1}}{x_{k+1}} & -1+\frac{\bar{z}_{k+1}}{x_{k+1}}  & \underline 0 \\
0 & 1-\frac{\bar{z}_{k+1}}{x_{k+1}}& \underline 1 \\
 \underline 1 &\underline 1 &  M^{k}_n(\mathbf{p}')  
\end{vmatrix}\\
& =\Big(1-\frac{z_{k+1}}{x_{k+1}}\Big)
 \begin{vmatrix}
 1-\frac{\bar{z}_{k+1}}{x_{k+1}}  & \underline 1 \\
\underline 1 &   M^{k}_n(\mathbf{p}') 
\end{vmatrix}
 +\Big(1-\frac{\bar{z}_{k+1}}{x_{k+1}}\Big)  
 \begin{vmatrix}
 0 & \underline 1 \\
\underline 1 &  M^{k}_n(\mathbf{p}') 
\end{vmatrix} \\
&=
\Big(1-\frac{z_{k+1}}{x_{k+1}}\Big)
 \begin{vmatrix}
 1 & \underline 1 \\
\underline 1 &  M^{k}_n(\mathbf{p}') 
\end{vmatrix}
+
\Big(1-\frac{z_{k+1}}{x_{k+1}}\Big)
\begin{vmatrix}
 -\frac{\bar{z}_{k+1}}{x_{k+1}}   & \underline 0 \\
\underline 1 &  M^{k}_n(\mathbf{p}') 
\end{vmatrix} \\
&  +\Big(1-\frac{\bar{z}_{k+1}}{x_{k+1}}\Big) 
 \begin{vmatrix}
 1 & \underline 1 \\
\underline 1 &  M^{k}_n(\mathbf{p}') 
\end{vmatrix} 
 +\Big(1-\frac{\bar{z}_{k+1}}{x_{k+1}}\Big) 
 \begin{vmatrix}
 -1 & \underline 0 \\
\underline 1 &  M^{k}_n(\mathbf{p}') 
\end{vmatrix} \\
&= \Big(2-\frac{2}{x_{k+1}}\cos \Big(\frac{2\pi i}{n} p_{k+1}\Big) \Big) \begin{vmatrix}
 1 & \underline 1 \\
\underline 1 &  M^{k}_n(\mathbf{p}') 
\end{vmatrix} + \Big(\frac{1}{x_{k+1}^2}-1 \Big) \det M^{k}_n(\mathbf{p}'),
\end{align*} 
Let $ \overbar{ M}^{k}_n(\mathbf{p}') $ be the matrix
$M^{k}_n(\mathbf{p}') $, where from each entry we subtract $1$. It is a block diagonal matrix with blocks of size $2$:
\[
\begin{pmatrix}
 -\frac{z_i }{x_i}  & -1 \\
-1 &   -\frac{\bar{z}_i}{x_i}
\end{pmatrix} ,
\]
for $1 \leq i \leq k$,  whose rows and columns correspond to the pairs of directed edges $(\mathbf{p}' ,  \pm \mathbf{v})$.
Hence,
\[
\begin{vmatrix}
 1 & \underline 1 \\
\underline 1 &  M^{k}_n(\mathbf{p}') 
\end{vmatrix}  = 
\begin{vmatrix}
 1 & \underline 1 \\
\underline 0 &  \overbar M^{k}_n(\mathbf{p}') 
\end{vmatrix}
= \det \overbar M^{k}_n(\mathbf{p}')  =\prod_{i=1}^k\Big(\frac{1}{x_i^2}-1\Big).
\]
Therefore, by the induction assumption,
\begin{align*}
& \det(\Id - \hat \Lambda^{k+1}_n(\mathbf{p}) ) \\
&=  2x_{k+1}\Big[x_{k+1} -\cos \Big(\frac{2\pi}{n} p_{k+1}\Big)\Big] \prod_{i=1}^{k}(1-x_i^2)+(1-x_{k+1}^2) \det(\Id - \hat \Lambda^{k}_n(\mathbf{p}) ) \\
&= \Bigg[1+2 \sum_{i=1}^{k+1}\frac{x_i \big(x_i- \cos\big(\frac{2\pi}{n} p_i  \big)\big) }{1-x_i^2}\Bigg]\prod_{i=1}^{k+1}(1-x_i^2).  \qedhere
\end{align*}
 \end{proof}

From all these considerations we obtain an exact formula for the partition function of the model on the torus.
\begin{corollary} \label{cor:partition_function}
The partition function of the model on $\graph^d_n$ with translation invariant weights $\mathbf{x}=(x_1,\ldots,x_d)$ satisfying
\begin{align} \label{eq:weightcond}
\sum_{i=1}^d \frac{x_i}{1+x_i} < \frac12
\end{align}
is
 \[
 Z_{\LS} = \prod_{\mathbf{p} \in \graph^d_n}\Bigg[1+2 \sum_{i=1}^{d}\frac{x_i \big(x_i- \cos\big(\frac{2\pi}{n} p_i  \big)\big) }{1-x_i^2}\Bigg]^{-\frac12}\prod_{i=1}^d(1-x_i^2)^{-\frac{n^d}{2}}.
 \]
\end{corollary}
\begin{proof}
We use \eqref{lem:blockdet} and Lemma~\ref{lem:blockdet}, and note that the determinants of all blocks are positive whenever \eqref{eq:weightcond} holds true.
\end{proof}

The \emph{free energy density} of the model is defined as minus the logarithm of the partition function divided by the ``volume''
(the number of edges):
\[
f(\mathbf{x}) = - \frac{\log Z_{\LS} }{ |E|}.
\]
As an easy consequence of the corollary above, we obtain that the limiting free energy density as $\graph^d_n$ approaches $\IZ^d$ 
is given by an explicit formula.
\begin{corollary} \label{cor:free_energy_density}
The free energy density of the model on $\graph^d_n$ with translation invariant weights $\mathbf{x}$ satisfying \eqref{eq:weightcond} in the thermodynamic limit
$\graph^d_n \nearrow \IZ^d$ is given by
\begin{align*}
 f(\mathbf{x}) \cdot 2d =  \sum_{i=1}^d\log(1-x_i^2) +  \frac1{(2\pi)^d} \int_{[0,2\pi]^d}\log\Bigg[1+2 \sum_{i=1}^{d}\frac{x_i \big(x_i- \cos \alpha_i\big) }{1-x_i^2}\Bigg] d\mathbf{\alpha}.
\end{align*}

\end{corollary}


Note that the logarithm in the integral diverges as $\alpha \to 0$ on the critical surface
\[
\sum_{i=1}^d \frac{x_i}{1+x_i} = \frac12.
\]

From now on, we will simplify the setting by considering only homogenous models with a single parameter $x$ such that $\mathbf{x}=(x,\ldots,x)$.
In this case, the critical point is $x_c = 1/(2d-1)$.
We now analyze the behavior of the singular part of the free energy density as $x \nearrow x_c$.
We will write $A(x) \sim B(x)$ if $c_1 B(x) + c_2 \leq A(x) \leq C_1 B(x) + C_2$,
for some constants $c_1,c_2,C_1,C_2$ (depending on $d$) as $x \nearrow x_c$.
\begin{corollary} \label{cor:singular_behavior}
Let $\delta=d/2$ if $d$ is even and $\delta=(d+1)/2$ if $d$ is odd. If $n<\delta$, then $d^n f(x)/dx^n$ stays finite as $x \nearrow x_c$.
If $d$ is even, then $d^{\delta}f(x)/dx^{\delta} \sim \log(x_c-x)$, and
if $d$ is odd, then $d^{\delta}f(x)/dx^{\delta} \sim (x_c-x)^{-1/2}$ as $x \nearrow x_c$. 
\end{corollary}

\begin{proof}
Letting $2xp(x)=(2d-1) x^2 - 2dx + 1$ (so that $p(x_c)=0$ and $\frac{p(x)}{x_c-x}\to const \neq 0$ as $x\nearrow x_c$), up to constants, one can write the singular part of $f$ as follows:
\begin{eqnarray*}
\lefteqn{
\int_{[0,2\pi]^d}\log \Big[ p(x) +  \sum_{i=1}^d (1 - \cos\alpha_i) \Big] d\mathbf{\alpha} } \\
& = & 2 \int_{[0,\pi]^d}\log \Big[ p(x) +  \sum_{i=1}^d (1 - \cos\alpha_i) \Big] d\mathbf{\alpha} \\
& \sim & \int_{[0,\pi/2]^d}\log \Big[ p(x) +  \sum_{i=1}^d (1 - \cos\alpha_i) \Big] d\mathbf{\alpha} \\
& = & \int_{[0,1]^d}\log \Big[ p(x) +  \sum_{i=1}^d y_i \Big] \prod_{i=1}^d \frac{dy_i}{\sqrt{y_i(2-y_i)}} \\
& \sim & \int_{[0,1]^d}\log \Big[ p(x) +  \sum_{i=1}^d y_i \Big] \prod_{i=1}^d y_i^{-1/2} dy_i \\
& = & 2^d \int_{[0,1]^d}\log \Big[ p(x) +  \sum_{i=1}^d z_i^2 \Big] \prod_{i=1}^d dz_i \\
& \sim & 2^d \int_{0}^{\sqrt{d}}\log \big[ p(x) +  r^2 \big] r^{d-1} dr \\
& = & 2^d d^{d/2-1} \log[p(x)+d] - \frac{2^{d+2} x}{d} \int_0^{\sqrt d} \frac{r^{d+1}}{p(x) +  r^2} dr .
\end{eqnarray*}

We see that the integral is convergent at $x=x_c$ for all $d\geq2$. However, taking $n$ derivatives of $f$ generates
a term containing the integral
\[
\int_0^{\sqrt d} \frac{r^{d-1} dr}{(p(x) +  r^2)^n}.
\]
At $x=x_c$, this integral is convergent if $2n<d$ and divergent if $2n \geq d$.

Writing $p=p(x)$, for $x$ sufficiently close to $x_c$, we have that
\begin{align*}
\int_0^{\sqrt d} \frac{r^{d-1} dr}{(p +  r^2)^n} =\int_{\sqrt p}^{\sqrt {d+p}} \frac{(s^2-p)^{\frac{d}{2}-1}}{s^{2n-1}}ds \leq \int_{\sqrt p}^{\sqrt {d+p}} s^{d-1-2n} ds
\end{align*}
and 
\begin{align*}
\int_0^{\sqrt d} \frac{r^{d-1} dr}{(p +  r^2)^n} & =\int_0^{\sqrt p} \frac{r^{d-1} dr}{(p +  r^2)^n} +\int_{\sqrt p}^{\sqrt d} \frac{r^{d-1} dr}{(p +  r^2)^n}\\ 
&\geq \frac1{(2p)^n}\int_0^{\sqrt p} r^{d-1} dr + \frac1{2^n}\int_{\sqrt p}^{\sqrt d} r^{d-1-2n} dr \\
&= \frac1{d2^n}p^{d/2-n}+ \frac1{2^n}\int_{\sqrt p}^{\sqrt d} r^{d-1-2n} dr.
\end{align*}
The last statement of the lemma follows taking $n=\delta$.
\end{proof}

\subsection{The one-point function}
In this section we compute the one-point function of the homogenous model on $\graph^d_n$ and $\IZ^d$.
We begin with a lemma which expresses it in terms of the Green's function. 
The result is proved by expressing the desired quantity in terms of a similar object in the soup of \emph{oriented loops}, and then 
repeating a classic proof of an analogous statement for general loop soups \cite{lejan11}. Let
\[
G_{\de,\dg} = [(\Id - \Lambda)^{-1}]_{\de,\dg}
\]
be the \emph{Green's function} for the non-backtracking random walk.
If $X$ is a random variable,
we will write $\langle X \rangle$ for its expectation.

\begin{lemma} \label{lem:onepointfunction} For any edge $\ue$,
\[
\langle N_{\LS}(\ue) \rangle= G_{\de,\de}-1,
\]
where $\de$ is any of the two oriented versions of $\ue$. As a consequence, 
\[
\Big \langle \frac{1}{|E|}  \sum_{\ue \in \uE} N_{\LS}(\ue) \Big \rangle=  \frac{1}{ 2 |E|} \Tr(\Id - \Lambda)^{-1}  -1
\]
\end{lemma}
\begin{proof}
Let $\vec \LS$ be the soup of unrooted but oriented loops with intensity $\frac12$, i.e.\ a Poisson point process 
with intensity measure $\frac12\vec \mu$, where $\vec \mu(\vec \Loop) = \mu(\Loop)$, and where $\vec \Loop$ is an oriented version of $\Loop$.  
For an oriented edge $\de$, let $N_{\vec \LS}(\de)$ be the number of times $\vec \LS$ visits $\de$. 
One has 
\[
N_{\LS}(\ue) \mathop{=}^{d} N_{\vec \LS}(\de) + N_{\vec \LS}(-\de),
\]
and, since the distribution of $\vec \LS$ is invariant under reversal of all loops,
\[
N_{\vec \LS}(\de)  \mathop{=}^{d} N_{\vec \LS}(-\de).
\]
Hence $\langle N_{\LS}(\ue) \rangle= 2 \langle N_{\vec \LS}(\de) \rangle$.

Fix an oriented edge $\de$. For $|t|\leq1$, let
\[
[\Lambda_{t}]_{\de_1,\de_2} = \begin{cases}
		x_{ \ue_1}(t\mathds{1}_{\{\de_1 =\de\}} + \mathds{1}_{\{\de_1 \neq \de\}} )& \text{if } h_{\de_1}=t_{\de_2} \text{ and } t_{\de_1}\neq h_{\de_2}, \\	
		0 & \text{otherwise},
		 \end{cases}
\]
and let $Z_t$ denote the partition function of the corresponding loop soup.
Using expression \eqref{eq:part_function} for the partition function and Lemma~\ref{lem:expdet}, we see that
\begin{align*}
\langle N_{\vec \LS}(\de) \rangle &  = \frac{d}{dt} \log Z_t \big|_{t=1} \\
& = \frac{\frac{d}{dt}  {\det}^{-\frac12}(\Id - \Lambda_t) \big|_{t=1} }{ {\det}^{-\frac12}(\Id - \Lambda)}  = 
-\frac12 \frac{\frac{d}{dt}  {\det}(\Id - \Lambda_t) \big|_{t=1} }{ {\det}(\Id - \Lambda)} \\
&= \frac12 \Tr\big[(\Id - \Lambda)^{-1}\frac{d}{dt} \Lambda_t \big] \big|_{t=1}
=  \frac12 \Tr\big[(\Id - \Lambda)^{-1} I_{\de} \Lambda \big] \\
& = \frac12 \Tr\big[ I_{\de} \sum_{n=1}^{\infty} \Lambda^n\big] 
=  \frac12 \Tr\big[ I_{\de} \big((\Id - \Lambda)^{-1}-\Id \big) \big] \\
&=  \frac12(G_{\de,\de} -1 ),
\end{align*}
where the fourth identity follows from Jacobi's formula for the derivative of a determinant, and where $[I_{\de}]_{\de_1,\de_2} = \mathds{1}_{\{\de_1=\de_2= \de\}}$. \qedhere
\end{proof}

As in the case of the partition function, using the above result which relates the one-point function to the underlying matrix, exact computations can be made for the homogenous model.
\begin{lemma} \label{lem:spectrum}
The eigenvalues of $\hat \Lambda^d_n(\mathbf{p})$ are $\pm x$ with multiplicity $d-1$, and 
\[
x \, \frac{2d-1}{a_{\mathbf{p}} \pm \sqrt{a^2_{\mathbf{p} }-2d +1}}
\]
with multiplicity $1$, where  $a_{\mathbf{p}} = \sum_{i=1}^d  \cos \big(\frac{2\pi}{n} p_i \big)$.
\end{lemma}
\begin{proof} The result follows from Lemma~\ref{lem:blockdet}.
\end{proof}
Note that it follows that the spectral radius of $\Lambda^d_n$ is equal to $(2d-1)x$ and is achieved by one of the multiplicity-one eigenvalues of $\hat \Lambda^d_n(\mathbf{p})$ for $\mathbf{p} =(0,\ldots,0)$.

\begin{corollary} \label{cor:onepointfunc}
For any edge $\ue$ of $\graph^d_n$,
\[
\langle N_{\LS}(\ue) \rangle =\frac{1}{ d n^d} \sum_{\mathbf{p} \in \graph^d_n} \frac{1-x\sum_{i=1}^d \cos\big(\frac{2\pi}{n} p_i \big)}{1+(2d-1)x^2-2x\sum_{i=1}^d  \cos\big(\frac{2\pi}{n} p_i \big)} +A(x),
\]
and hence, in the thermodynamic limit,
\begin{align*}
\lim_{\graph^d_n \nearrow \IZ^d}\langle N_{\LS}(\ue) \rangle &= \frac{1}{ d(2\pi)^d }  \int_{[0,2\pi]^d} \frac{1-x\sum_{i=1}^d  \cos\alpha_i}{1+(2d-1)x^2-{2x}\sum_{i=1}^d  \cos \alpha_i} d\alpha +A(x),
\end{align*}
where $A(x)= \frac{d-1}{d}(1-x^2)^{-1} -1$ is smooth on $(0,1)$.
\end{corollary}
\begin{proof}
Let $\sigma^d_n({\mathbf{p}})$ be the spectrum of $\hat \Lambda^d_n({\mathbf{p}})$.
Using Lemmas~\ref{lem:onepointfunction} and~\ref{lem:spectrum}, we have
\begin{align*}
\langle N_{\LS}(\ue) \rangle+1& = \Big \langle \frac{1}{|E|}  \sum_{\ue' \in \uE} N_{\LS}(\ue') \Big \rangle +1=  \frac{1}{ 2 |E|} \Tr(\Id - \Lambda^d_n)^{-1}   \\
&= \frac{1}{2dn^d} \sum_{\mathbf{p} \in \graph^d_n} \sum_{\lambda \in \sigma^d_n({\mathbf{p}})} (1-\lambda)^{-1} \\
&=  \frac{d-1}{d}(1-x^2)^{-1}+\frac{1}{ d n^d} \sum_{\mathbf{p} \in \graph^d_n} \frac{1- xa_{\mathbf{p}}}{1+(2d-1)x^2 -2xa_{\mathbf{p}}}. \qedhere
\end{align*}
\end{proof}

\begin{corollary} \label{cor:onepointfuncdivergnce}
As $x \nearrow x_c=   1/(2d-1)$, $ \lim_{\graph^d_n \nearrow \IZ^d}\langle N_{\LS}(\ue) \rangle$
stays bounded for $d\geq3$, and diverges logarithmically for $d=2$.
\end{corollary}

\begin{proof}
Corollary \ref{cor:onepointfunc} and a computation analogous to the one in the proof of Corollary \ref{cor:singular_behavior} yield
\[
\lim_{\graph^d_n \nearrow \IZ^d}\langle N_{\LS}(\ue) \rangle \sim 2^d \int_0^{\sqrt d} \big[ p(x) + r^2 \big]^{-1} r^{d-1} dr .
\]
The integral is convergent for all $d\geq3$; for $d=2$, one has
\[
4 \int_0^{\sqrt 2} \frac{r dr}{[p(x) + r^2]} = 2 \log\Big( \frac{p(x)+2}{p(x)} \Big),
\]
which diverges logarithmically as $x \nearrow x_c$.
\end{proof}

\section{The distribution of $N_{\LS}(\ue)$} \label{sec:distribution}
In this section we compute the probability generating function of $N_{\LS}(\ue)$ and then use the result to prove a limit theorem for
the two-dimensional edge-occupation field.

For $ |z| \leq 1$, let 
\[
p_{\ue}(z) = \sum_{n=0}^{\infty} \IP(N_{\LS}(\ue) = n )z^n
\]
be the probability generating function of $N_{\LS}(\ue)$.
For a directed edge $\de$, let $F_{\de}$ be the partition function of directed loops $\dWalk$ rooted at $\de$ which do not visit $-\de$ and such that $\dWalk_i=\de$ only for $i=1$ and $i=|\Walk|+1$ (that is, the sum over all such loops of the weights of the loops).
Let $F'_{\de}$ be the partition function of walks $\dWalk$ rooted at $\de$ such that $\dWalk_i = \de$ only for $i=1$ and $\dWalk_i=-\de$ only for $i=|\Walk|+1$.

\begin{lemma} \label{lem:probgenfunc}
\[
p_{\ue}(z) = \bigg( \frac{1- zF_{-\de} - z^2 F'_{\de}F'_{-\de}}{1- F_{-\de} -  F'_{\de}F'_{-\de}} \bigg)^{-1/2}
\]
\end{lemma}
\begin{proof}
Let $Z_z$ be the partition function of the soup of oriented loops with intensity measure $\frac12 \vec\mu_z$,
where $ \vec \mu_z(\vec \Loop) = \mu(\Loop) z^{N_{\Loop}(\ue)}$. 
One has
\[
p_{\ue}(z) = \frac{Z_z}{Z_1} = \exp\bigg(\frac{1}{2} \sum_{\vec \Loop:\ \ue \in \vec \Loop}  \vec \mu_z (\vec \Loop) \bigg) /  \exp\bigg(\frac{1}{2} \sum_{\vec \Loop:\ \ue \in \vec \Loop}  \vec \mu_1 (\vec \Loop) \bigg).
\]
Using the identity
\[
 \sum_{\vec \Loop:\ \ue \in \vec \Loop}  \vec \mu_z (\vec \Loop) =  \sum_{\vec \Loop:\ \de \in \vec \Loop}  \vec \mu_z (\vec \Loop) + 
 \sum_{\vec \Loop: -\de \in \vec \Loop,\ \de \notin \vec  \Loop}  \vec \mu_z (\vec \Loop)
\]
and a well-known fact (see, for instance, Lemma 9.3.2 in~\cite{LawlerLimic} or Lemma 4 in \cite{Lis} for a proof in the non-backtracking case), one gets 
\begin{align} \label{eq:probgenfunc}
 \exp\bigg(\frac{1}{2} \sum_{\vec \Loop:\ \ue \in \vec \Loop}  \vec \mu_z (\vec \Loop) \bigg)  = (1-  \tilde G^z_{\de,\de})^{-1/2}(1-zF_{-\de})^{-1/2},
\end{align}
where $\tilde G^z_{\de,\de}$ is the partition function of oriented loops $\dWalk$ rooted at $\de$ such that $\dWalk_i=\de$ only for $i=1$ and $i=|\Walk|+1$, weighted with the weight $x^z(\dWalk) = x(\dWalk) z^{N_{\dWalk}(\ue)}$. Splitting these loops according to their first and last visit to $-\de$, one has
\[
\tilde G^z_{\de,\de} = zF'_{\de}(1-zF_{-\de})^{-1}zF'_{-\de},
\]
which together with \eqref{eq:probgenfunc} finishes the proof.
\end{proof}

Contrary to dimension three and higher, in two dimensions the edge-occupation field is not defined at the critical point
(see, for example, Corollary~\ref{cor:onepointfuncdivergnce}). Nevertheless, in $\graph^2_n$ for any $n$, including $n=\infty$,
one can use Theorem~\ref{lem:probgenfunc} and the next corollary to prove a limit theorem, as $x \nearrow 1/3$, for the field
normalized by its expectation.

\begin{corollary} \label{cor:onepointagain}
\[
\langle N_{\LS}(e) \rangle = \frac12 \frac{F_{-\de} + 2 F'_{\de}F'_{-\de}}{1- F_{-\de} -  F'_{\de}F'_{-\de}}
\]
\end{corollary}
\begin{proof}
We use the fact that $\langle N_{\LS}(e) \rangle = \frac{d}{d z} p_{\ue}(z)\big |_{z=1}$ and Lemma~\ref{lem:probgenfunc}.
\end{proof}

\begin{theorem}
Fix $n$ (possibly $n=\infty$) and consider the loop soup $\LS$ in $\graph^2_n$. Then, for any edge $\ue$, as $x \nearrow 1/3$,
${N_{\LS}(e)}/{\langle N_{\LS}(e) \rangle}$ converges in distribution to the square of the standard normal distribution.
\end{theorem}
\begin{proof}
We will use L\'evy's continuity theorem. Let $\varphi(t)$ be the characteristic function of ${N_{\LS}(e)}/{\langle N_{\LS}(e) \rangle}$, 
and let $\varepsilon =1- F_{\de} -  F'_{\de}F'_{-\de}$ and $C =F_{-\de} + 2 F'_{\de}F'_{-\de}$. 
From Corollaries~\ref{cor:onepointfunc} and~\ref{cor:onepointagain}, we can deduce that, as $x \nearrow 1/3$, $\varepsilon \to 0$ and $F_{-\de}$, $F'_{\de}$ and $F'_{-\de}$ remain bounded.
Hence, by Lemma~\ref{lem:probgenfunc} and Corollary~\ref{cor:onepointagain},  
\begin{align*}
\lim_{x \nearrow 1/3}\varphi(t) &= \lim_{\varepsilon \to 0}  \bigg(1- \frac{(e^{2it \varepsilon/C }-1)F_{-\de} + (e^{4it\varepsilon/C}-1) F'_{\de}F'_{-\de} }{\varepsilon} \bigg)^{-1/2}  \\
&= (1-2it)^{-1/2},
\end{align*}
which is the characteristic function of the square of the standard normal distribution.
\end{proof}

\section{The two-point function} \label{sec:two-point_function}
In this section we show the existence of a subcritical regime with exponential decay of correlations.
We let
\[
\bar N_{\LS}(\ue) = N_{\LS}(\ue) - \langle N_{\LS}(\ue) \rangle
\]
be the truncated two-point function.
We denote by $N_{\Loop}(e)$ the number of visits of the loop $\Loop$ to $e \in E$,
and write $e \in \Loop$ if $\Loop$ visits $e$ at least once.
\begin{lemma} \label{lem:truncated_twopoint}
For any pair of edges $e,g$, we have that
\[
\langle  \bar N_{\LS}(e) \bar N_{\LS}(g)\rangle = \sum_{\Loop: \ e,g \in \Loop} N_{\Loop}(e) N_{\Loop}(g)\mu(\Loop).
\]
\end{lemma}
\begin{proof}
We will use the obvious identity
\[
N_{\LS}(e) = \sum_{\Loop : \ e \in \Loop, \ g \not\in \Loop} (\#\Loop) N_{\Loop}(e) 
+ \sum_{\Loop : \ e, g \in \Loop} (\#\Loop) N_{\Loop}(e),
\]
where $\# \Loop$ is the multiplicity of $\Loop$ in $\LS$.
Using the fact that $\LS$ is a Poisson point process, this gives
\begin{align*}
& \langle  \bar N_{\LS}(\ue) \bar N_{\LS}(\ug)\rangle = \langle  N_{\LS}(\ue)  N_{\LS}(g)\rangle - \langle  N_{\LS}(\ue) \rangle \langle N_{\LS}(\ug)\rangle \\
& = \Big\langle \Big( \sum_{\Loop: e,g  \in \Loop} (\#\Loop) N_{\Loop}(\ue) \Big) \Big( \sum_{\Loop: e,g  \in \Loop} (\#\Loop) N_{\Loop}(\ug) \Big) \Big\rangle \\
& -  \Big\langle \sum_{\Loop: e,g  \in \Loop} (\#\Loop) N_{\Loop}(\ue) \Big\rangle \Big\langle \sum_{\Loop: e,g  \in \Loop} (\#\Loop) N_{\Loop}(\ug) \Big\rangle \\
& = \sum_{\Loop,\Loop':\ e,g \in \Loop,\Loop'} N_{\Loop}(\ue) N_{\Loop'}(\ug) \left\langle (\#\Loop)(\#\Loop') \right\rangle
- \sum_{\Loop,\Loop':\ e,g \in \Loop,\Loop'} N_{\Loop}(\ue) N_{\Loop'}(\ug) \left\langle \#\Loop \right\rangle \left\langle \#\Loop' \right\rangle \\
& = \sum_{\Loop:\ e,g  \in \Loop} N_{\Loop}(\ue) N_{\Loop}(\ug) \left( \left\langle (\#\Loop)^2 \right\rangle - \left\langle \#\Loop \right\rangle^2 \right) = \sum_{\Loop:\ e,g  \in \Loop} N_{\Loop}(\ue) N_{\Loop}(\ug) \mu(\Loop) . \qedhere
\end{align*}
\end{proof}
On a regular lattice where each vertex has $2d$ nearest neighbors, the number of rooted, non-backtracking walks of length $k$
is bounded above by $(2d-1)^k$. Lemma \ref{lem:truncated_twopoint} then makes it clear that for $x<1/(2d-1)$ one should
expect exponential decay of the truncated two-point function, identifying $x<1/(2d-1)$ as the subcritical regime and $x=1/(2d-1)$
as the critical point.

We now express the truncated two-point function in terms of the Green's function, and use the expression to give a proof of
exponential decay in the subcritical regime.
\begin{lemma} \label{lem:twopoint}
For any pair of edges $e,g$, we have that
\[
2\langle  \bar N_{\LS}(\ue) \bar N_{\LS}(\ug)\rangle = G_{\de,\dg}G_{\dg,\de} + G_{-\de,\dg}G_{\dg,-\de} + G_{-\de,-\dg}G_{-\dg,-\de} + G_{\de,-\dg}G_{-\dg,\de} \, .
\]
\end{lemma}
\begin{proof}
As for the one-point function, we consider the soup of oriented loops $\vec \LS$.
We have
\[
\langle  N_{\LS}(e)  N_{\LS}(g)\rangle = \big\langle  \big(N_{\vec \LS}(\de)+ N_{\vec \LS}(-\de)\big)\big( N_{\vec \LS}(\dg)+ N_{\vec \LS}(-\dg) \big)\big\rangle.
\]
Let us compute $\langle N_{\vec \LS}(\de)  N_{\vec \LS}(\dg)  \rangle$. For $|s|,|t|\leq1$, let $Z_{s,t}$ be the
partition function of the soup of oriented loops with intensity measure $\frac12 \vec \mu_{s,t}$, where

\[
\vec \mu_{s,t}(\vec \Loop) = \vec \mu(\vec\Loop) s^{N_{\vec\Loop}(\de)}t^{N_{\vec\Loop}(\dg)},
\]
with $N_{\vec\Loop}(\de)$ being the number of visits of $\vec\Loop$ to $\de$.
Using the form of the partition function (see \eqref{eq:part_function}), one has
\begin{align*}
& \langle  N_{\vec\LS}(\de)  N_{\vec\LS}(\dg)\rangle=\frac{ \frac{\partial }{\partial s} \frac{\partial }{\partial t} Z_{s,t}\big|_{s,t=1}}{Z_{1,1}}  = \frac{\frac{\partial }{\partial s} \frac{\partial }{\partial t} \exp ( \log Z_{s,t}) \big|_{s=t=1}}{Z_{1,1}} \\
&= \frac{\partial }{\partial s} \frac{\partial }{\partial t}\log Z_{s,t}\big|_{s,t=1} + \frac{\partial }{\partial s} \log Z_{s,t} \frac{\partial }{\partial t} \log Z_{s,t}\big|_{s,t=1} \\
&= \frac12 \frac{\partial }{\partial s}G_{\de,\de}(s)\big|_{s=1} + \frac14 \langle N_{\LS}(\ue) \rangle \langle N_{\LS}(g)\rangle,
\end{align*}
where we used Lemma~\ref{lem:onepointfunction}, and where $G^s$ is the Green's function for the non-backtracking walk with weight $x_s( \Walk) = x( \Walk) s^{N_{ \Walk}(\dg)}$.
Let $ \tilde G_{\dg,\dg}$ be the partition function of oriented loops $\dWalk$ rooted at $\dg$ such that $\dWalk_i=\dg$ only for $i=1$ and $i=|\Walk|+1$
(that is, the sum over all such loops of the weights of the loops). Splitting loops that visit $\dg$ multiple times into loops that visit $\dg$ only once, one can see that $G^s_{\dg,\dg}= (1-s\tilde G_{\dg,\dg})^{-1}-1$,
and hence $G^s_{\de,\de}= C s(1-s\tilde G_{\dg,\dg})^{-1}$, where $C$ does not depend on $s$ (the additional factor $s$ comes from the first visit of the walk to $\dg$).  Note that only loops visiting
$\dg$ survive the differentiation $\frac{\partial}{\partial s}$, and that
\[
 \frac{\partial }{\partial s}  s(1-s\tilde G_{\dg,\dg})^{-1} \big|_{s=1} = (1-\tilde G_{\dg,\dg})^{-2} = G^2_{\dg,\dg}.
\]
Decomposing the rooted loops starting at $\de$ according to their first and last visit to $\dg$ and using the identity above we obtain that $\frac{\partial }{\partial s}G^s_{\de,\de}\big|_{s=1} =G_{\de,\dg}G_{\dg,\de}$,
which finishes the proof.
\end{proof}

\begin{corollary}[Exponential decay] \label{cor:exponential_decay}
For the homogeneous model on the torus $\graph^d_n$ with $x_e=x<1/(2d-1)$, one has
\[
\langle  \bar N_{\LS}(\ue) \bar N_{\LS}(\ug)\rangle  \leq  2 \frac{(x(2d-1))^{2d(\ue,\ug)}}{(1-x(2d-1))^2},
\]
where $d(\ue,\ug)$ is the graph distance between $\ue$ and $\ug$.
\end{corollary}
\begin{proof}
By Lemma~\ref{lem:twopoint}, it is enough to prove that $G_{\de,\dg} \leq \frac{(x(2d-1))^{d(\ue,\ug)}}{1-x(2d-1)}$.
The induced $L^1$ norm of the transition matrix is $\| \Lambda^d_n \|_1 = x(2d-1)$. Hence,
\begin{align*}
G_{\de,\dg} &= [(\Id -\Lambda_n^d  )^{-1}]_{\de,\dg} = [(\Lambda_n^d)^{d(\ue,\ug)}(\Id -\Lambda_n^d  )^{-1}]_{\de,\dg} \\
& \leq \|(\Lambda_n^d)^{d(\ue,\ug)} (\Id -\Lambda_n^d  )^{-1}\|_{1}\leq \|\Lambda_n^d\|_1^{d(\ue,\ug)}(1-\| \Lambda_n^d\|_1)^{-1}. \qedhere
\end{align*}
\end{proof}

\section{A reflection positive spin model} \label{sec:spin_model}
In this section we consider the loop soup on the square lattice, defined as the graph with vertices $\IZ^2$ and edges between nearest-neighbor vertices. The dual graph $(\IZ^2)^*$ is again a square lattice. To each vertex $v^*$ of the dual graph $(\IZ^2)^*$, we associate a $(\pm 1)$-valued (spin)
variable $\sigma_{v^*}$.

We define a spin model on the vertices of $(\IZ^2)^*$ by taking a loop soup $\LS$ in $\IZ^2$ and assigning, to each dual vertex $v^*$, spin
\begin{align} \label{eq:spin}
\sigma_{v^*} = \exp\bigg(\frac{i}{2} \sum_{\Loop \in \LS} \theta_{\Loop}(v^*)\bigg) ,
\end{align}
where $\theta_{\Loop}(v^*)$ is the winding angle (a multiple of $\pm2\pi$) of loop $\Loop$ around $v^*$
(and $i$ here denotes the imaginary unit). Here we assume that there are only finitely many loops surrounding any dual vertex. This is true for any $x <1/3$.

We will now show that this spin model is reflection positive. (See \cite{Biskup} for more information on the concept and use of reflection
positivity in the context of lattice spin models. The question of reflexion positivity in the loop soup context is addressed in Chapter 9 of \cite{lejan11}.)
$\IZ^2$ has a natural reflection symmetry along any line $l$ going through a set of dual vertices. Such a line splits $\IZ^2$ in two halves, $\IZ^2_{+}$
and $\IZ^2_{-}$. We also split accordingly the dual graph $(\IZ^2)^*$ in two halves, $(\IZ^2_{+})^*$ and $(\IZ^2_{-})^*$, such that
$(\IZ^2_{+})^* \cap (\IZ^2_{-})^* = V^*_l$, where $V^*_l$ is the set of vertices of $(\IZ^2)^*$ that lie on $l$.

Let ${\mathcal F}^+$ (respectively, ${\mathcal F}^-$) denote the set of all functions of the spin variables
$(\sigma_{v^*})_{v^* \in (\IZ^2_{+})^*}$ (respectively, $(\sigma_{v^*})_{v^* \in (\IZ^2_{-})^*}$).
Let $\vartheta$ denote the reflection operator, $\vartheta: {\mathcal F}^{\pm} \to {\mathcal F}^{\mp}$, whose action on
spins is given by $\vartheta(\sigma_{v^*}) = \sigma_{\vartheta(v^*)}$, and let $\E$ denote expectation with respect to the
loop soup.
\begin{lemma}[Reflection positivity] \label{lem:reflection_positivity}
For all functions $f,g \in {\mathcal F}^+$, $\E(f \vartheta g) = \E(g \vartheta f)$ and $\E(f \vartheta f) \geq 0$.
\end{lemma}
\begin{proof}
Take $v^* \in (\IZ^2_{\pm})^*$ and draw a path in $(\IZ^2_{\pm})^*$ from $v^* $ to infinity.  Let $n_{v^*}$ be the number of loops and arcs of the soup in 
$\IZ^2_{\pm}$ which cross the path an odd number of times. One can see that $\sigma_{v^*} =(-1)^{n_{v^*}}$, and hence the spin model in $(\IZ^2_{\pm})^*$ is a function of the 
soup restricted to $\IZ^2_{\pm}$.

Let $E_l$ denote the edges of $\IZ^2$ crossed by the reflection line $l$.
Theorem \ref{thm:spatial_markov} implies therefore that $f$ and $\vartheta g$ are independent conditional on $N_S |_{E_l}$.
Using reflection symmetry, this gives $\E(f \vartheta g | N_S |_{E_l}) = \E(f | N_S |_{E_l}) \E(\vartheta g | N_S |_{E_l})$.
The proof of the lemma follows immediately from the law of total expectation.
\end{proof}


Defining the spin model on the square lattice is convenient to express reflection positivity, but one can of course define the
same model on other two-dimensional graphs. Consider, for example, the spin model defined via \eqref{eq:spin} on a finite
subset of the square or hexagonal lattice with zero boundary condition $\xi$ and equal edge-weights $x_e=x \; \forall e \in E$.
We define the (unnormalized) magnetization field as $\sum_{v^*} \sigma_{v^*} \delta_{v^*}$, where $\delta_{v^*}$
denotes the Dirac delta at $v^*$. 
We conjecture that, if $x=1/3$ in the case of the square lattice or $x=1/2$ in the case of the hexagonal lattice, the
magnetization field, when properly rescaled, has a continuum scaling limit which gives rise to one of the conformal
fields discussed in \cite{CGK} and \cite{CamLis}, namely, the winding field with $\lambda=1/2$ and $\beta=\pi$
in the language of \cite{CGK}. (Note that, if $x<1/3$ on the square lattice or $x<1/2$ on the hexagonal lattice, the
truncated two-point function can be easily shown to decays exponentially by arguments similar to those in the proof of
Lemma \ref{lem:truncated_twopoint}.)

Some evidence in favor of this conjecture comes from \cite{BJVL}. In that paper, the authors consider the scaling limit of a
quantity that can be interpreted as the one-point function of a spin model defined via \eqref{eq:spin}, but using walks without
the non-backtracking condition. They express the limit in terms of the Brownian loop soup \cite{LawWer} in a way that is consistent
with our conjecture. We expect that the non-backtracking condition does not change the scaling limit of the loop soup, so the same
computation should apply to the one-point function of our spin model.

Our conjecture and Lemma \ref{lem:reflection_positivity} suggest that the winding field of \cite{CGK} and \cite{CamLis} with
$\lambda=1/2$ and $\beta=\pi$ is reflection positive and satisfies the Osterwalder-Schrader axioms of Euclidean field theory
\cite{osterwalder1973}.

We also conjecture that, both in the case of the square lattice with $x=1/3$ and in the case of the hexagonal lattice with $x=1/2$,
in the continuum scaling limit, the collection of boundaries of the spin clusters converges to CLE$_4$, the Conformal Loop Ensemble
with parameter $\kappa=4$.

\medskip

\noindent{\bf Acknowledgments.} The first author acknowledges the support of VIDI grant 639.032.916 of the Netherlands Organization
for Scientific Research. He also thanks Chuck Newman and Roberto Fernandez for useful suggestions. Part of the research was conducted
while the second author was at Brown University. The second author also thanks VU University Amsterdam for its hospitality during two
visits while the paper was completed. Both authors are grateful to Yves Le Jan for several useful conversations and remarks on a previous
draft of the paper, and for pointing their attention to references \cite{lejan15} and \cite{lejan16}, and to the Ihara zeta function.
They also thank Wendelin Werner for a useful discussion on the spatial Markov property of loop soups.

\bibliographystyle{amsplain}
\bibliography{rlf}

\end{document}